\documentclass[11pt,reqno]{amsart}
\usepackage{amsmath,amssymb, mathrsfs, amsfonts}
\usepackage{graphicx, color, amscd}
\usepackage{comment}
\usepackage{slashed}
\newcommand{\RN}[1]{%
  \textup{\uppercase\expandafter{\romannumeral#1}}%
}
\newcommand{\pl}{\partial}

\newcommand{\rw}{\rightarrow}

\newcommand{\Lb}{\underline{L}}
\newcommand{\chib}{\underline{\chi}}

\newcommand{\R}{\mathbb{R}}

\newcommand{\lt}{\left}
\newcommand{\rt}{\right}
\newcommand{\na}{\nabla}
\renewcommand{\hat}{\widehat}
\newcommand{\cb}{\underline{c}}
\newcommand{\nb}{\underline{n}}
\newcommand{\tr}{\mbox{tr}}
\newcommand{\Fb}{\underline{F}}

\newtheorem{theorem}{Theorem}
\newtheorem{proposition}[theorem]{Proposition}
\newtheorem{prop}[theorem]{Proposition}
\newtheorem{corollary}[theorem]{Corollary}
\newtheorem{lemma}[theorem]{Lemma}

\theoremstyle{definition}
\newtheorem{definition}[theorem]{Definition}
\newtheorem{remark}[theorem]{Remark}

\title[Evolution of conserved quantities at null infinity]{Evolution of angular momentum and center of mass at null infinity}
\author[P.-N. Chen, J. Keller, M.-T. Wang, Y.-K. Wang, and S.-T. Yau]{Po-Ning Chen, Jordan Keller, \\ Mu-Tao Wang, Ye-Kai Wang, and Shing-Tung Yau}
\numberwithin{theorem}{section}
\numberwithin{equation}{section}

\begin{document}

\begin{abstract}
We study how conserved quantities such as angular momentum and center of mass evolve with respect to the retarded time at null infinity, which is described in terms of a Bondi-Sachs coordinate
system. These evolution formulae complement the classical Bondi mass loss formula for gravitational radiation. They are further expressed in terms of the potentials of the shear and news tensors.
The consequences that follow from these formulae are (1) Supertranslation invariance of the fluxes of the CWY conserved quantities.  
(2) A conservation law of angular momentum  \`a la Christodoulou. (3) A duality paradigm for null infinity. In particular, the supertranslation invariance distinguishes the CWY angular momentum and center of mass from the classical definitions.
\end{abstract}

\thanks{P.-N. Chen is supported by NSF grant DMS-1308164 and Simons Foundation collaboration grant \#584785, M.-T. Wang is supported by NSF grant DMS-1810856, Y.-K. Wang is supported by MOST Taiwan grant 107-2115-M-006-001-MY2, 109-2628-M-006-001 -MY3 and S.-T. Yau is supported by NSF grants  PHY-0714648 and DMS-1308244. The authors would like to thank the National Center for Theoretical Sciences at National Taiwan University where part of this research was carried out. This material is based upon work supported by the National Science Foundation under Grant No. DMS-1810856.} 

\maketitle

\section{Introduction}
In this article, we study the evolution of angular momentum and center of mass at null infinity of asymptotically flat vacuum spacetimes. These evolution formulae complement the classical Bondi mass loss formula for gravitational radiations. We are particularly interested in the total flux of angular momentum and center of mass.  

For a good notion of conserved quantities, one expects that the total flux is independent of the choice of coordinate systems. However, as indicated by Penrose \cite{Penrose2}, the notion of ``angular momentum carried away by gravitational radiation"
can be shifted by supertranslations, an infinite dimensional symmetry at null infinity. Such ambiguity has been a crucial obstacle to a clear understanding of conserved quantities at null infinity. In this article, we consider both the classical and the Chen-Wang-Yau (CWY) \cite{CWY3} definitions for angular momentum and center of mass at null infinity. A key result is the supertranslation invariance of the flux of the 
CWY angular momentum and center of mass. This invariance distinguishes the CWY definitions from the classical definitions.

Consider the future null infinity $\mathscr{I}^+$  of an asymptotically flat spacetime, which is described in terms of a Bondi-Sachs coordinate system. $\mathscr{I}^+$ is identified with $I \times S^2$, where $I\subset (-\infty, +\infty)$ is an interval parametrized by the retarded time $u$ and $S^2$ is the standard unit 2-sphere equipped with the standard round metric $\sigma_{AB}$.  Let $m$ denote the mass aspect, $N_A$ the angular momentum aspect, $C_{AB}$ the shear tensor, and $N_{AB}$ the news tensor of $\mathscr{I}^+$.  One can view $m$ as a smooth function, $N_A$ a smooth one-form, and $C_{AB}$ and $N_{AB}$ smooth symmetric traceless 2-tensors (with respect to $\sigma_{AB}$) on $S^2$ that depend on $u$. In particular, $\partial_u C_{AB}=N_{AB}$. See a brief description of $\mathscr{I}^+$ in the Bondi-Sachs coordinates and the definitions of these quantities in Section 2.

All integrals in this paper on the sphere are taken over the standard two-sphere $S^2$ with the standard round metric $\sigma_{AB}$. 
We take the standard formulae for energy and linear momentum:
\begin{equation}\begin{split} E &= \int_{S^2} 2m\\
 P^k &= \int_{S^2} 2m \tilde{X}^k, k=1, 2, 3\end{split}\end{equation}
where $\tilde{X}^k, k=1, 2 ,3$ are the standard coordinate functions on $\R^3$ restricted to the unit sphere $S^2$. 

Furthermore, we consider the classical angular momentum 
\begin{equation}\label{angmom}
\tilde J^k=\int_{S^2} \epsilon^{AB}\nabla_{B}\tilde{X}^{k}[N_A-\frac{1}{4}C_{A}^{\,\,\,\,D}\nabla^B C_{DB}], 
\end{equation}
and the classical center of mass 
\begin{equation}\label{com}
\tilde C^k = \int_{S^2} \nabla^{A}\tilde{X}^{k}[N_A- u \nabla_A m-\frac{1}{4}C_{A}^{\,\,\,\,D}\nabla^B C_{DB} - \frac{1}{16}\nabla_{A}(C_{DE}C^{DE})],
\end{equation}
where $\nabla_A$ denotes the covariant derivative with respect to $\sigma_{AB}$, and $\epsilon_{AB}$ denotes  the volume form of $\sigma_{AB}$ and $k=1, 2, 3$. The indexes are raised, lowered, and contracted with respect to $\sigma_{AB}$. Our definition is that of Dray-Streubel \cite{DS}. See Section III.B of Flanagan-Nichols \cite{Flanagan-Nichols} for details.

%Our definition is taken from \cite[(6.117)]{CJK}. It coincides with the definition of Dray-Streubel \cite{DS}. (We learned this from Section III.B of Flanagan-Nichols \cite{Flanagan-Nichols}.)

\begin{remark}
In the above definitions of conserved quantities, we omit the constant $\frac{1}{8\pi}$.
\end{remark}

Furthermore, we consider the CWY angular momentum $J^k$ and center of mass $C^k$ as the limits of the CWY quasi-local angular momentum and center of mass \cite{CWY3,CWY4} on $\mathscr{I}^+$  evaluated in \cite{KWY}.
\[ J^{k} = \int_{S^2}  \epsilon^{AB} \na_B \tilde X^k \lt( N_{A} -\frac{1}{4}C_{AB}\nabla_{D}C^{DB} - c\nabla_{A}m   \rt)
\]
\[
\begin{split}
 C^k &= \int_{S^2} \nabla^{A}\tilde{X}^{k} \Bigg[ N_A - u\na_A m- \frac{1}{4} C_{AB} \na_D C^{DB} - \frac{1}{16} \na_A \lt( C^{DE}C_{DE}\rt) \\
&\hspace{2.5cm} - c \na_A m  + 2 \epsilon_{AB} (\na^B \cb) m \Bigg]\\
&\quad  +  \int_{S^2}  3 \tilde X^k cm - \frac{1}{4} \tilde X^k \na_A \Fb^{AB} \na^D \Fb_{DB}
\end{split}
\]
where $c$ and $\underline{c}$ are the potentials of $C_{AB}$, as given in \eqref{CAB} and $\Fb_{AB} = \frac{1}{2} (\epsilon_{AD}\na_B\na^D \cb + \epsilon_{BD}\na_A\na^D \cb).$ 
For definiteness, the potentials are assumed to be supported in the $\ell\ge 2$ modes.

In Theorem 11 and Theorem 16 of \cite{KWY}, it is shown that $J^k$ and $\frac{ C^k}{E}$ are the limit of the Chen-Wang-Yau quasi-local angular momentum and center of mass (omitting constant $1/8\pi$) under the zero linear momentum assumption 
\begin{align}\label{zero momentum}
\int_{S^2} m(u,x) \tilde X^i =0.
\end{align}
The CWY angular momentum and center of mass modify the classical definitions as follows:
\begin{equation}
\begin{split}
 J^{k} = & \tilde J^k - \int_{S^2}  \epsilon^{AB} \na_B \tilde X^k   c\nabla_{A}m  \\
 C^{k}= & \tilde C^k +\int_{S^2} \nabla^{A}\tilde{X}^{k} \lt ( - c \na_A m  + 2 \epsilon_{AB} (\na^B \cb) m \rt)\\
 & \quad +  \int_{S^2}  3 \tilde X^k cm - \frac{1}{4} \tilde X^k \na_A \Fb^{AB} \na^D \Fb_{DB}
\end{split}
\end{equation}
The  correction terms come from solving the optimal isometric embedding equation in the theory of Wang-Yau quasilocal mass \cite{Wang-Yau1,Wang-Yau2} and are non-local. They provide the reference terms  that are critical in the Hamiltonian approach of defining conserved quantities. See \cite{KS} for a definition of angular momentum in the context of perturbations of Kerr, in which the referencing is achieved by the uniformization theorem.

The ten conserved quantities $(E, P^k, \tilde J^k, \tilde C^k)$, or $(E, P^k, J^k, C^k)$, are functions on $I$ that depend on the retarded time $u$. We compute the derivatives of these conserved quantities with respect to $u$. In particular, for the classical angular momentum and center of mass, we obtain

 \begin{theorem}\label{evolution1} The classical angular momentum $\tilde J^k$ and center of mass $\tilde C^k$, $k=1,2,3$ evolve according to the following:
\begin{align}
\partial_u \tilde J^k &=\frac{1}{4}\int_{S^2} \lt[ \epsilon^{AE} \nabla_E \tilde{X}^k  (C_{AB}\nabla_{D}N^{BD} -N_{AB}\nabla_{D}C^{BD}) +  \tilde{X}^k  \epsilon^{AB} (C_{A}^{\,\,\,\,D} N_{DB}) \rt],
\label{dujk1}\\
\partial_u \tilde C^k &= \frac{1}{4}\int_{S^2} \lt[ \nabla^{A}\tilde{X}^{k} \lt( \frac{u}{2}\nabla_A|N|^2 + C_{AB}\nabla_{D}N^{BD} - N_{AB}\nabla_{D}C^{BD} \rt) \rt]. \label{duck1}
\end{align}
\end{theorem}

The evolution formulae \eqref{dujk1}  and \eqref{duck1} can be further expressed in  terms of the potentials  of $C_{AB}$ and $N_{AB}$:
\begin{theorem}\label{evolution2_bracket}
Suppose $c$ and $\underline{c}$ are the potentials of $C_{AB}$ and $n$ and $\underline{n}$ are the potentials of  $N_{AB}$, as given in \eqref{CAB} and \eqref{NAB}, then 
\begin{equation}\label{dubracket}\begin{split}\partial_u \tilde J^k&=\frac{1}{8} \int_{S^2} \tilde{X}^k ( [ c, \Delta(\Delta+2) n]_1+[ \underline{c}, \Delta(\Delta+2) \underline{n}]_1)\\
\partial_u\tilde  C^k&=\frac{1}{8}\int_{S^2} \tilde{X}^k \Big ( u [((\Delta+2)n)^2 +((\Delta+2) \underline{n})^2- 4\epsilon^{AB} \nabla_A  n \nabla_B (\Delta+2)\underline{n}] \\
& \qquad  \qquad  \quad + [ (\Delta+2) c, (\Delta+2) n]_2+[ (\Delta+2) \underline{c}, (\Delta+2) \underline{n}]_2\Big),\end{split}\end{equation} where $[\cdot, \cdot]_1$ is the Poisson bracket on $S^2$
defined in \eqref{bracket1} and $[\cdot, \cdot]_2$ is another bracket  on $S^2$ defined in \eqref{bracket2}. 
\end{theorem}

The Bondi-Metzner-Sachs (BMS) group acts on $\mathscr{I}^+$. It includes supertranslations which we will review in further details in Section 5. The ambiguity of supertranslations has presented an essential difficulty to understanding the structure of $\mathscr{I}^+$ since the 1960s. Among $(m, N_A, C_{AB}, N_{AB})$, only $N_{AB}$ is a supertranslation invariant quantity.  
It is natural to ask whether total flux of angular momentum is invariant under a supertranslation. For the classical angular momentum, we prove that
\begin{corollary}[Theorem \ref{supertranslation invariance of total flux}]
Suppose $\mathscr{I}^+$ extends from $u=-\infty$ to $u=+\infty$ and the news tensor decays as
\[ N_{AB}( u,x) = O(|u|^{-1-\varepsilon}) \mbox{ as } u \rw \pm\infty,
\]
then the total flux of the classical angular momentum $\tilde J^k$ is supertranslation invariant if and only if
\begin{align}
\lim_{u \rw +\infty} m(u,x) - \lim_{u \rw -\infty}m(u,x)
\end{align}
is supported in the $l\le 1$ modes.
\end{corollary}
In particular, if $\lim_{u \rw +\infty} m(u,x) - \lim_{u \rw -\infty}m(u,x)$ contains $l\ge 2$ modes, the total flux of the classical angular momentum will depend on the supertranslation. This demonstrates how the total flux of the classical angular momentum can be shifted by a supertranslation. On the other hand, we show that the CWY angular momentum is free of such supertranslation ambiguity. 

\begin{theorem}[Theorem \ref{total_flux_CWY_angular}]
Suppose the news tensor decays as
\[ N_{AB}( u,x) = O(|u|^{-1-\varepsilon}) \mbox{ as } u \rw \pm\infty.
\]
Then the total flux of $J^k$ is  supertranslation invariant.
\end{theorem}
\begin{remark}
In the above statement,  supertranslation invariant means that it is equivariant under ordinary ($l=1$) translation and is invariant under higher mode ($l \ge 2$) of the supertranslation. See the statement of Theorem \ref{total_flux_CWY_angular} for further details.
\end{remark}

We also show that the invariance under supertranslation distinguishes the CWY center of mass from the classical center of mass. Indeed, the total flux of the classical center of mass is invariant under supertranslation if and only if $\lim_{u \rw +\infty} m(u,x) - \lim_{u \rw -\infty}m(u,x)$ is a constant function on $S^2$. On the other hand, the total flux of the CWY center of mass is always supertranslation invariant. See the statement of Theorem \ref{total_flux_CWY_COM}.

Next, we show that if a spacetime admits a Bondi-Sachs coordinate system with vanishing news tensor, then $(E, P^k, J^k, C^k)$ are constant (independent of the retarded time $u$) and supertranslation invariant. See the statement of Theorem \ref{invariant_CWY_vanish_news} for further details.

While our focus is on the study of angular momentum and center of mass in a Bondi-Sachs coordinate system, we show that the evolution formulae for the classical angular momentum can be carried over to the framework of the stability of Minkowski spacetime \cite{CK} if we take Rizzi's definition of angular momentum \cite{Rizzi, Rizzi_thesis}. This provides a conservation law of angular momentum that complements the conservation law for linear momentum of Christodoulou \cite[Equation (13)]{Christo1991}.

Another natural consequence of \eqref{dubracket} is a duality paradigm among sets of null infinity data $(m, N_A, C_{AB}, N_{AB})$, through replacing the potentials $(c,\cb,n,\nb)$ by $(-\cb, c, -\nb, n)$.

\begin{corollary}\label{duality}
Given a set of null infinity data $(m, N_A, C_{AB}, N_{AB})$ defined on $[u_1,u_2] \times S^2$,  there exists a dual set of null infinity data  $({m}^{*}, {N}^{*}_A, {C}^{*}_{AB}, {N}^{*}_{AB})$ that has the same (classical) energy, linear momentum, angular momentum, and center-of-mass.
\end{corollary}

These are dual sets of null infinity data that are indistinguishable in terms of the classical conserved quantities.

The paper is organized as follows. In Section 2, we introduce the definitions and integration by parts formulae used throughout the paper. The flux of classical conserved quantities is computed in Section 3 and is rewritten in terms of the potentials in Section 4. The aforementioned consequences of flux formulae are presented in Section 5 to Section 7. In the last section, we consider the case of quadrupole moment radiation. With the future theoretical and numerical investigation in mind, we express the flux formulae in terms of the spherical harmonics expansion of potentials explicitly.
%{\color{blue} Subsequently, we will investigate relations between the evolution of energy and angular momentum and between the evolution of linear momentum and center-of-mass.
%\begin{question} Suppose $\partial_u P^i=0$, can we find any relation between $\partial_u E$ and $\partial_u J$? 
%\end{question}}

\section{Background information}
In this section, we describe the Bondi-Sachs coordinate system and recall several useful formulae for functions and tensors on $S^2$.
\subsection{Bondi-Sachs coordinates}
In terms of a Bondi-Sachs coordinate system $(u, r,  x^2, x^3)$, near $\mathscr{I}^+$ of a vacuum spacetime, the metric takes the form
\begin{equation}\label{spacetime_metric}g_{\alpha\beta}dx^\alpha dx^\beta= -UV du^2-2U dudr+r^2 h_{AB}(dx^A+W^A du)(dx^B+W^B du).\end{equation} 

The index conventions here are $\alpha, \beta=0,1, 2, 3$, $A, B=2, 3$, and $u=x^0, r=x^1$. See \cite{BVM, MW} for more details of the construction of the coordinate system. 

The metric coefficients $U, V, h_{AB}, W^A$  of \eqref{spacetime_metric} depend on $u, r, \theta, \phi$, but $\det h_{AB}$ is independent of $u$ and $r$. These gauge conditions thus reduce the number of metric coefficients of a Bondi-Sachs coordinate system to six (there are only two independent components in $h_{AB}$). On the other hand, the boundary conditions $U\rightarrow 1$, $V\rightarrow 1$, $W^A\rightarrow 0$, $h_{AB}\rightarrow \sigma_{AB}$ are imposed as $r\rightarrow \infty$ (such boundary conditions may not be satisfied in a radiative spacetime). Here $\sigma_{AB}$ denotes a standard round metric on $S^2$. The special gauge choice implies a hierarchy among the vacuum Einstein equations, see \cite{MW, HPS}.

 Assuming the outgoing radiation condition \cite{BVM, Sachs, MW}, the boundary condition and the vacuum Einstein equation imply that as $r\rightarrow \infty$, all metric coefficients can be expanded in inverse integral powers of $r$.\footnote{The outgoing radiation condition assumes the traceless part of the $r^{-2}$ term in the expansion of $h_{AB}$ is zero. The presence of this traceless term will lead to a logarithmic term in the expansions of $W^A$ and $V$. Spacetimes with metrics which admit an expansion in terms of $r^{-j}\log^i r$ are called ``polyhomogeneous" and are studied in \cite{CMS}. They do not obey the outgoing radiation condition or the peeling theorem \cite{VK}, but they do appear as perturbations of the Minkowski spacetime by the work of Christodoulou-Klainerman \cite{CK}.} In particular (see Chrusciel-Jezierski-Kijowski \cite[(5.98)-(5.100)]{CJK} for example), 
\[\begin{split} U&=1 - \frac{1}{16r^2} |C|^2 + O(r^{-3}),\\
V&=1-\frac{2m}{r}+ \frac{1}{r^2}\lt( \frac{1}{3}\na^A N_A + \frac{1}{4} \na^A C_{AB} \na_D C^{BD} + \frac{1}{16}|C|^2 \rt) + O(r^{-3}),\\
W^A&= \frac{1}{2r^2} \na_B C^{AB} + \frac{1}{r^3} \lt( \frac{2}{3}N^A - \frac{1}{16} \na^A |C|^2 -\frac{1}{2} C^{AB} \na^D C_{BD} \rt) + O(r^{-4}),\\
h_{AB}&={\sigma}_{AB}+\frac{C_{AB}}{r}+ \frac{1}{4r^2} |C|^2 \sigma_{AB} + O(r^{-3})\end{split} \] where  $m=m(u, x^A)$ is the mass aspect, $N_A = N_A(u,x^A)$ is the angular aspect and $C_{AB}=C_{AB}(u, x^A)$ is the shear tensor of this Bondi-Sachs coordinate system. Note that our convention of angular momentum aspect differs from that of Chrusciel-Jezierski-Kijowski \cite{CJK}, $N_A =-3 N_{A(CJK)}$. Here we take norm, raise and lower indices of tensors with respect to the metric $\sigma_{AB}$. We also define the {\it news tensor} $N_{AB} = \pl_u C_{AB}$.

\subsection{Integral formulae on 2-sphere} Let $\sigma_{AB}$ be the standard round metric on $S^2$ with respect to which the indexes of tensors are raised or lowered. Let $\nabla_A$ be covariant derivative with respect to $\sigma_{AB}$. Let $\epsilon_{AB}$ be the volume form. The following identity 
\begin{align}\label{epsilon_square}
\epsilon_{AB} \epsilon_{CD} = \sigma_{AC}\sigma_{BD} - \sigma_{AD}\sigma_{BC}
\end{align}
and its contraction
\begin{align}\label{contract_epsilon_square}
\epsilon_{AB}\epsilon^A_{\;\;C} = \sigma_{BC}
\end{align}
will be used frequently.
  
The curvature formula on $S^2$ gives  \[\nabla_A \nabla_B \nabla_C u- \nabla_B \nabla_A \nabla_C u=\sigma_{AC} \nabla_B u-\sigma_{BC} \nabla_A u\] for a smooth function $u$ on $S^2$. In particular, we have
\begin{equation}\label{commutation}\begin{split}\nabla_D \nabla^D\nabla_A u&=\nabla_A (\Delta+1)u \\ \epsilon^{AB}\nabla_A \nabla_B \nabla_C u&=\epsilon_C^{\;\; B}\nabla_B u.\\
 \end{split}\end{equation}

Let $\tilde{X}^k, k=1, 2, 3$ be the restriction to $S^2$ of the standard coordinate functions in $\mathbb{R}^3$. It is well-known that they are eigenfunctions for $\sigma_{AB}$:
\[ \Delta \tilde X^k = -2 \tilde X^k.\] $\tilde X^k$ also satisfies the Hessian equation
\begin{equation}\label{hessian_X_k} \nabla_A \nabla_B \tilde{X}^k=-\tilde{X}^k \sigma_{AB}.\end{equation}
In general, an eigenfunction $f$ with
\begin{equation}\label{mode_def} \Delta f = - \ell(\ell+1)f\end{equation}
is said to be of mode $\ell$. 
We need the following integration by parts lemma:
\begin{lemma}\label{same_mode}
Suppose $u$ and $v$ are smooth functions on $S^2$ of mode $m$ and $n$ respectively. Then
\[\int_{S^2} \tilde{X}^k \epsilon^{AB} \nabla_A u \nabla_B v=0 \] unless $m=n$. 
\end{lemma}

\begin{proof}
Integrating by parts, we obtain  \[\int_{S^2} \tilde{X}^k \epsilon^{AB} \nabla_A u \nabla_B v= \int_{S^2} (Y^A \nabla_A v)u,\] where $Y^A=\epsilon^{AB}\nabla_B \tilde{X}^k$ is 
a rotation Killing field. Since $\Delta$ commutes with $Y^A \nabla_A $,  $Y^A \nabla_A v $ is of the same mode as $v$. 
\end{proof} 
The following integrating by parts formulae will be useful in the later sections.
\begin{lemma} For any smooth functions $u, v$ on $S^2$, we have
\begin{equation}\label{ibp1} \int_{S^2} \tilde{X}^k \epsilon^{AB} \nabla_A (\Delta u) \nabla_B v=\int_{S^2} \tilde{X}^k \epsilon^{AB} \nabla_A u \nabla_B (\Delta v)\end{equation}
\begin{equation}\label{ibp2} \int_{S^2} \tilde{X}^k \epsilon^{AB}\nabla_A \nabla_D u\nabla_B \nabla^D v=-\int_{S^2} \tilde{X}^k \epsilon^{AB} \nabla_A u\nabla_B (\Delta+2) v.\end{equation}
\end{lemma}

\begin{proof} We prove the second formula and the first formula follows similarly.  Integrating by parts the left hand side, we obtain 

\[  -\int_{S^2} \nabla_D \tilde{X}^k \epsilon^{AB}\nabla_A u\nabla_B \nabla^D v- \int_{S^2} \tilde{X}^k \epsilon^{AB}\nabla_A u \nabla_D \nabla_B \nabla^D v\]

Integrating the first term by parts again, we obtain
\[  \int_{S^2} \nabla_B \nabla_D \tilde{X}^k \epsilon^{AB}\nabla_A u \nabla^D v- \int_{S^2} \tilde{X}^k \epsilon^{AB}\nabla_A u \nabla_D \nabla_B \nabla^D v\]

By \eqref{commutation}, this is equal to 
\[    - \int_{S^2} \tilde{X}^k \epsilon^{AB}\nabla_A u \nabla_B v- \int_{S^2} \tilde{X}^k \epsilon^{AB}\nabla_A u \nabla_B (\Delta+1) v.\]

\end{proof}

\begin{lemma} \label{quadratic_integral} For any smooth function $u$ on $S^2$, we have
\[ \begin{split} \int_{S^2} [2\nabla_A\nabla_Bu \nabla^A\nabla^B u- (\Delta u)^2]&=\int_{S^2} u \Delta(\Delta+2)u\\
\int_{S^2} \tilde{X}^i [2\nabla_A\nabla_Bu \nabla^A\nabla^B u- (\Delta u)^2]&=\int_{S^2} \tilde{X}^i [(\Delta+2)u]^2.\\
\end{split}\]
\end{lemma}
\begin{proof}
We use the following formulae in the derivation

\[\begin{split}
\Delta|\nabla u|^2&=2|\nabla^2u|^2+2\nabla u\cdot \nabla (\Delta+1)u\\
\Delta (u^2)&=2|\nabla u|^2+2 u\Delta u\\
\Delta(u\Delta u)&=(\Delta u)^2+2\nabla u\cdot \nabla(\Delta u)+u \Delta^2 u.
\end{split}\]

We prove the second formula and the first one follows similarly. Integrating by parts twice gives

\[\int_{S^2} \tilde{X}^i \nabla_A\nabla_Bu \nabla^A\nabla^B u=\int_{S^2}  u \nabla^A\nabla^B(\tilde{X}^i \nabla_A\nabla_Bu) \]

We compute 
\[\begin{split}
      &\nabla^A\nabla^B(\tilde{X}^i \nabla_A\nabla_Bu)\\
=&(\nabla^A\nabla^B \tilde{X}^i) \nabla_A\nabla_Bu+2\nabla^B \tilde{X}^i \nabla^A \nabla_A\nabla_Bu+\tilde{X}^i \nabla^A\nabla^B \nabla_A\nabla_Bu\\
=&-\tilde{X}^i \Delta u+2\nabla^B \tilde{X}^i \nabla_B (\Delta+1)u+\tilde{X}^i \Delta (\Delta+1)u\\
=&\tilde{X}^i \Delta^2 u+2\nabla^B \tilde{X}^i \nabla_B (\Delta+1)u   \end{split},\] where we use $\nabla^A \nabla_A\nabla_Bu=\nabla_B (\Delta+1)u$. 

On the other hand, we have the identity:

\[ 2\nabla^B u\nabla_Bv=\Delta (uv)- u\Delta v-v\Delta u\] and thus 
\[2\nabla^B \tilde{X}^i \nabla_B (\Delta+1)u=\Delta (\tilde{X}^i  (\Delta+1)u)-\tilde{X}^i \Delta(\Delta+1)u+2\tilde{X}^i  (\Delta+1)u.\]
Putting all together gives:
\[\begin{split} &\int_{S^2} \tilde{X}^i \nabla_A\nabla_Bu \nabla^A\nabla^B u\\
=&\int_{S^2} \tilde{X}^i \Delta^2 u+\int u  [\Delta (\tilde{X}^i  (\Delta+1)u)-\tilde{X}^i \Delta(\Delta+1)u+2\tilde{X}^i  (\Delta+1)u]\\
=&\int_{S^2} \tilde{X}^i [(\Delta u)^2+2u\Delta u+2u^2]\end{split}.\]

Therefore, 
\[\int_{S^2} \tilde{X}^i [2\nabla_A\nabla_Bu \nabla^A\nabla^B u- (\Delta u)^2]=\int_{S^2} \tilde{X}^i [(\Delta u )^2+4u \Delta u+4u^2]=\int_{S^2} \tilde{X}^i[(\Delta+2)u]^2.\]

\end{proof}

\subsection{Closed and Co-closed Decomposition}

In this subsection, we consider symmetric traceless 2-tensors $C_{AB}$ and $N_{AB}$ on $S^2$ with the decomposition (see \cite[Appendix B]{KWY} for a derivation)
\begin{equation}\label{CAB}C_{AB}=\nabla_A\nabla_B c-\frac{1}{2} \sigma_{AB} \Delta c+\frac{1}{2}(\epsilon_A^{\,\,\,\, E} \nabla_E \nabla_B \underline{c}+\epsilon_B^{\,\,\,\, E} \nabla_E \nabla_A \underline{c})\end{equation}
\begin{equation}\label{NAB}N_{AB}=\nabla_A\nabla_B n-\frac{1}{2} \sigma_{AB} \Delta n+\frac{1}{2}(\epsilon_A^{\,\,\,\, E} \nabla_E \nabla_B \underline{n}+\epsilon_B^{\,\,\,\, E} \nabla_E \nabla_A \underline{n})\end{equation} for smooth functions $c, \underline{c}, n, \underline{n}$ on $S^2$ that are referred as potentials of $C_{AB}$ and $N_{AB}$. The potentials  are unique up to their $0$ and $1$ mode. In the case we consider when $C_{AB}$ and $N_{AB}$ depend on $u$, all $c, \underline{c}, n, \underline{n}$ depend on $u$ as well. 

\begin{proposition}
Closed and co-closed parts of a symmetric traceless 2-tensors on $S^2$ are dual to each other in the following sense. 
\begin{enumerate}
\item Denote the space of symmetric traceless 2-tensors on $S^2$ by $\hat{\mbox{Sym}}$. Then the map $\varepsilon_2: \hat{\mbox{Sym}} \rw \hat{\mbox{Sym}}, \varepsilon_2(C_{AB}) = \epsilon_A^{\;\;D} C_{DB}$ satisfies 
\begin{align}
&\varepsilon_2(\nabla_A\nabla_B c-\frac{1}{2} \sigma_{AB} \Delta c) = \frac{1}{2}(\epsilon_A^{\,\,\,\, E} \nabla_E \nabla_B c +\epsilon_B^{\,\,\,\, E} \nabla_E \nabla_A c), \label{closed_to_coclosed}\\
&\varepsilon_2 \lt( \frac{1}{2}(\epsilon_A^{\,\,\,\, E} \nabla_E \nabla_B \cb +\epsilon_B^{\,\,\,\, E} \nabla_E \nabla_A \cb)\rt) = -\na_A\na_B \cb + \frac{1}{2} \sigma_{AB} \Delta \cb. \label{coclosed_to_closed} 
\end{align}
\item The following identity holds for symmetric traceless 2-tensors
\begin{align}
\epsilon_D^{\;\;B} \na^D C_{BA} = \epsilon_A^{\;\;D} \na^B C_{BD}. \label{epsilon_divergence}
\end{align}
In other words, we have a commutative diagram of isomorphisms
\begin{align*}
\begin{CD}
\hat{\mbox{Sym}} @>\varepsilon_2>> \hat{\mbox{Sym}}\\
@VV \mbox{div} V @VV \mbox{div}V\\
\Lambda^1 @> * >> \Lambda^1,
\end{CD}
\end{align*}
where $\Lambda^1$ denotes the space of 1-forms and $(*\omega)_A = \epsilon_A^{\;\;B} \omega_B$ is the Hodge star on 1-forms.
\end{enumerate}
\end{proposition}
\begin{proof}
We use \eqref{epsilon_square} and \eqref{contract_epsilon_square} in the derivation. Since $\epsilon^{AB} \varepsilon_2(C_{AB})=0$ and $\sigma^{AB} \varepsilon_2(C_{AB})=0$, $\varepsilon_2(C_{AB})$ is symmetric and traceless. In particular, 
\begin{align}
\epsilon_A^{\;\;D} C_{DB} = \frac{1}{2} (\epsilon_A^{\;\;D} C_{DB} + \epsilon_B^{\;\;D} C_{DA})
\end{align}
and \eqref{closed_to_coclosed} and \eqref{coclosed_to_closed} follow by direct computation.

To verify \eqref{epsilon_divergence}, note that both sides are equal to $\na^D C_{DE}$ after contracted with $\epsilon^A_{\;\;E}$.
\end{proof}

In the following two lemmas, we express several integrals involving the shear tensor and the news tensor in terms of their potentials. These formulae will help us to derive Theorem \ref{evolution2_bracket} from Theorem \ref{evolution1}.

\begin{lemma} Suppose $Y^A$ is either $\nabla^A \tilde{X}^k$ or $\epsilon^{AB}\nabla_B \tilde{X}^k$, and $C_{AB}$ and $N_{AB}$ are given by \eqref{CAB} and \eqref{NAB}, then

\begin{equation}\label{YCN}\begin{split}&\int_{S^2} Y^A C_{AB}\nabla_D N^{BD}\\
=&-\frac{1}{4}\int_{S^2} Y^A [(\Delta+2)n \nabla_A (\Delta+2)c +(\Delta+2)\underline{n} \nabla_A (\Delta+2)\underline{c}]\\
&+\frac{1}{4}\int_{S^2} Y^A \epsilon_A^{\,\,\,\,D} [\nabla_D ((\Delta+2)c) (\Delta+2)\underline{n}-\nabla_D ((\Delta+2)\underline{c}) (\Delta+2)n]\\
\end{split}\end{equation}

\end{lemma}

\begin{proof} First of all, note that 
\[\nabla^B Y^A C_{AB}=0, \epsilon^{BD}\nabla_D Y^A C_{AB}=0.\] From  
\[ \nabla_D N^{BD}=\frac{1}{2} \nabla^B (\Delta+2)n+\frac{1}{2}\epsilon^{BD}\nabla_D (\Delta+2)\underline{n}, \] we integrate by parts to get 
\[\int_{S^2} Y^A C_{AB}\nabla_D N^{BD}=-\frac{1}{2} \int_{S^2} Y^A       ( \nabla^B C_{AB}   (\Delta+2) n+\epsilon^{BD} \nabla_D C_{AB} (\Delta+2)\underline{n} ).  \]
By \eqref{epsilon_divergence} \[ \epsilon^{DB}\nabla_{D}C_{BA} = \epsilon_{AD}\nabla_{B}C^{BD}\]
and $\nabla^B C_{BD}=\frac{1}{2} \nabla^D (\Delta+2)c+\frac{1}{2}\epsilon^{BD}\nabla_B (\Delta+2)\underline{c}$, we obtain the desired formula.

%\[\begin{split}\int Y^A C_{AB}\nabla_D N^{BD}
%=&-\frac{1}{4}\int Y^A [(\Delta+2)n \nabla_A (\Delta+2)c +(\Delta+2)\underline{n} \nabla_A (\Delta+2)\underline{c}]\\
%&+\frac{1}{4}\int Y^A \epsilon_A^{\,\,\,\,D} [\nabla_D ((\Delta+2)c) (\Delta+2)\underline{n}-\nabla_D ((\Delta+2)\underline{c}) (\Delta+2)n] \end{split}\]

\end{proof}
The above generalizes the integral identities derived in \cite[(65), (66)]{KWY}:
\[ \int_{S^2} Y^{A} F_{A}^{B}\nabla^{D}F_{DB} = 0,\]
\[ \int_{S^2} Y^{A} \underline{F}_{A}^{B}\nabla^{D}\underline{F}_{DB} = 0\]
for $Y^A = \epsilon^{AB} \na_B \tilde X^k$.

Skew-symmetrizing \eqref{YCN}, we obtain: 
\begin{equation}\label{YCN_bracket}\begin{split}
&\int_{S^2} Y^A (C_{AB}\nabla_D N^{BD}-N_{AB}\nabla_D C^{BD})\\
=&\frac{1}{4}\int_{S^2} Y^A [(\Delta+2)c \nabla_A (\Delta+2)n-(\Delta+2)n \nabla_A (\Delta+2)c]\\
&+\frac{1}{4}\int_{S^2} Y^A [(\Delta+2) \underline{c} \nabla_A (\Delta+2)\underline{n}-(\Delta+2)\underline{n} \nabla_A (\Delta+2)\underline{c}]\\
&+\frac{1}{4}\int_{S^2} Y^A \epsilon_A^{\,\,\,\,D} \nabla_D [(\Delta+2)c(\Delta+2)\underline{n}-(\Delta+2)\underline{c} (\Delta+2)n].\\
 \end{split}\end{equation}

Next we prove

\begin{lemma}\label{news_square_integral}
\[\begin{split} \int_{S^2} N_{AB} N^{AB}&=\frac{1}{2} \int_{S^2}  n \Delta (\Delta+2)n +\underline{n} \Delta (\Delta+2) \underline{n} \\
\int_{S^2} \tilde{X}^k N_{AB} N^{AB}&=\frac{1}{2} \int_{S^2} \tilde{X}^k \Big [((\Delta+2)n)^2 +((\Delta+2) \underline{n})^2- 4\epsilon^{AB} \nabla_A  n \nabla_B (\Delta+2)\underline{n}\Big]. \end{split}\]
\end{lemma}

\begin{proof}
Using the formula $\epsilon_A^{\,\,\,\, C} \epsilon^{BD}=\delta_A^{\,\,\,\,B} \sigma^{CD}-\delta_A^{\,\,\,\,D}\sigma^{CB}$ and $\epsilon^{AB}\epsilon_{AE}=\sigma^{BE}$, we compute that 
\begin{equation}\label{news_square} \begin{split} N_{AB}N^{AB}&= \nabla_A\nabla_B {n} \nabla^A\nabla^B  {n}-\frac{1}{2} (\Delta {n})^2+   \nabla_A\nabla_B \underline{n} \nabla^A\nabla^B \underline{n}-\frac{1}{2} (\Delta \underline{n})^2\\
&+2\epsilon^{AC}\nabla_A \nabla_B n \nabla_C\nabla^B \underline{n}
\end{split}\end{equation}
Integrating by parts yields
\[
\begin{split}
\int_{S^2} \epsilon^{AC}\nabla_A \nabla_B n \nabla_C\nabla^B \underline{n}=&-\int_{S^2} \epsilon^{AC} \nabla^B \nabla_A\nabla_B n\nabla_C \underline{n}\\
=&-\int_{S^2} \epsilon^{AC} \nabla_A (\Delta+1)n \nabla_C \underline{n}=0
\end{split}
\]

The first formula now follows from the first formula in Lemma \ref{quadratic_integral}. The second formula follows from the second and third formula in Lemma \ref{quadratic_integral}.
\end{proof}

The second formula of Lemma \ref{news_square_integral} can be polarized and we obtain

\begin{equation}\label{polarize}\begin{split} \int_{S^2} \tilde{X}^k C_{AB} N^{AB}
=&\frac{1}{2} \int_{S^2} \tilde{X}^k \Big [(\Delta+2) c (\Delta+2)n +(\Delta+2)\underline{c} (\Delta+2) \underline{n} \\
& \quad  - 2\epsilon^{AB} (\nabla_A  c \nabla_B (\Delta+2)\underline{n}+\nabla_A  n \nabla_B (\Delta+2)\underline{c}) \Big] \end{split}
\end{equation}

\section{Evolution of Conserved quantities}
In this section, we compute the evolution of the classical angular momentum and center of mass. 
These formulae will be used to calculate the total flux of the conserved quantities.

Let's first review the evolution of the metric under the Einstein equation. It is well-known (see \cite[(5.102)]{CJK} for example) that the evolution of the mass aspect function is given by 
\begin{equation}\label{m_aspect_evol}\partial_u m=-\frac{1}{8}N_{AB} N^{AB}+ \frac{1}{4} \nabla^A\nabla^B N_{AB}.
\end{equation}

The modified mass aspect function $\hat{m}$ is defined to be \cite{MTW}

\begin{equation}\label{modified_mass}\hat{m}=m-\frac{1}{4} \nabla^A\nabla^B C_{AB}=m-\frac{1}{8} \Delta (\Delta+2) c\end{equation} and satisfies 
\begin{equation}\label{evol_modified_mass}\partial_u \hat{m}=-\frac{1}{8}N_{AB} N^{AB}.\end{equation}

Therefore,
\[\begin{split} \partial_u E&=- \frac{1}{4}\int_{S^2} N_{AB} N^{AB} \\
\partial_u P^k&= - \frac{1}{4}\int_{S^2} \tilde{X}^k   N_{AB} N^{AB}      , k=1, 2, 3.\\\end{split}\]

We also recall the evolution of $N_A$ (see \cite[(5.103)]{CJK} for example):
\[ \begin{split}\partial_u N_A &= \na_A m -\frac{1}{4}\nabla^D(\nabla_D \nabla^E C_{EA}-\nabla_A \nabla^E C_{ED}) \\
&\quad +\frac{1}{4}\nabla_A(C_{BE} N^{BE})-\frac{1}{4}\nabla_B (C^{BD} N_{DA})+\frac{1}{2} C_{AB}\nabla_D N^{DB}.\end{split}\]

The formula can be rewritten in the following form:

\begin{proposition}\label{AM_aspect_evol}
The angular momentum aspect $N_A$ evolves according to
\begin{equation}\label{duNA}\begin{split}\partial_u N_A
=& \na_A m + \frac{1}{4} \epsilon_{AB} \nabla^B( \epsilon^{PQ} \nabla_P \nabla^E C_{EQ})+\frac{1}{8}\nabla_A(C_{BE} N^{BE})\\
&+\frac{1}{8} \epsilon_{AB} \nabla^B (\epsilon^{PQ} C_P^{\,\,\, E}N_{EQ} )     +\frac{1}{2} C_{AB}\nabla_D N^{DB}.\end{split}\end{equation}
\end{proposition}

\begin{proof} We rewrite the terms $-\frac{1}{4}\nabla^D(\nabla_D \nabla^E C_{EA}-\nabla_A \nabla^E C_{ED})$ and $-\frac{1}{4}\nabla_B (C^{BD} N_{DA})$. First we check the following identity directly:
\[ \epsilon_{AB} \nabla^B( \epsilon^{PQ} \nabla_P \nabla^E C_{EQ})=-\nabla^D(\nabla_D \nabla^E C_{EA}-\nabla_A \nabla^E C_{ED}).\]

As for the term $C_{B}^{\,\,\, D}N_{DA}$,  we use the following general formulae for symmetric traceless 2-tensors on the 2-sphere:
\[ C_{B}^{\,\,\, D}N_{DA}+N_B^{\,\,\,D} C_{DA}=(C_{DE} N^{DE})\sigma_{AB}\]
\[C_B^{\,\,\,D}N_{DA}-N_B^D C_{DA}=-(\epsilon^{PQ} C_P^{\,\,\,E} N_{EQ})\epsilon_{AB}\]

Therefore, \[2C_{B}^{\,\,\, D}N_{DA}=(C_{DE} N^{DE})\sigma_{AB}-(\epsilon^{PQ} C_P^{\,\,\,E} N_{EQ})\epsilon_{AB}.\]
\end{proof}

Equation \eqref{duNA} is indeed equivalent to equation (4) on page 48 of \cite{ChristoMG9}. We apply \eqref{duNA} to derive the evolution of the classical angular momentum and center of mass.

\begin{theorem}[Theorem \ref{evolution1}]
The classical angular momentum and center of mass evolve according to the following:

\begin{equation}\label{dujk}
\partial_u \tilde J^k =\frac{1}{4}\int_{S^2} \lt[ \epsilon^{AE} \nabla_E \tilde{X}^k  (C_{AB}\nabla_{D}N^{BD} -N_{AB}\nabla_{D}C^{BD}) +  \tilde{X}^k  \epsilon^{AB} (C_{A}^{\,\,\,\,D} N_{DB}) \rt], 
\end{equation}
\begin{equation}\label{duck}
\partial_u \tilde C^k = \frac{1}{4}\int_{S^2} \lt[ \nabla^{A}\tilde{X}^{k} \lt(  \frac{u}{2} \nabla_A |N|^2+ C_{AB}\nabla_{D}N^{BD} - N_{AB}\nabla_{D}C^{BD} \rt) \rt],
\end{equation}
where $k=1,2,3.$
\end{theorem}
\begin{proof}
By \eqref{angmom},
\[ \partial_u \tilde J^k=\int_{S^2}    \epsilon^{AB}\nabla_{B}\tilde{X}^{k}[\partial_u N_A-\frac{1}{4}\partial_u(C_{A}^{\,\,\,\,D}\nabla^B C_{DB})]. \]
First, we deal with the term $\frac{1}{4} \epsilon_{AB} \nabla^B( \epsilon^{PQ} \nabla_P \nabla^E C_{EQ})$ on the right hand side of \eqref{duNA} and claim that 
\begin{equation}\label{Y^A first term of duN_A} \int_{S^2} Y^A \epsilon_{AB} \nabla^B( \epsilon^{PQ} \nabla_P \nabla^E C_{EQ})=0\end{equation} for $Y^A=\nabla^A \tilde{X}^k$ or $\epsilon^{AB}\nabla_B \tilde{X}^k$. Integrating by parts, the integral becomes 
\[\int_{S^2} \epsilon_{AB} \nabla^A Y^B (\epsilon^{PQ} \nabla_P \nabla^E C_{EQ}).\] Since $\epsilon^{PQ} \nabla_P \nabla^E C_{EQ}=-\frac{1}{2} \Delta(\Delta+2) \underline{c}$ and 
$\epsilon_{AB} \nabla^A Y^B$ is zero or $2 \tilde X^k$, the integral vanishes. 

Hence, we obtain
\[
\begin{split}
   &\partial_u \tilde J^k\\
= & \int_{S^2}    \epsilon^{AB}\nabla_{B}\tilde{X}^{k} \lt[ \frac{1}{8} \epsilon_{AE} \nabla^E (\epsilon^{PQ} C_P^{\,\,\, E}N_{EQ} )     +\frac{1}{2} C_{AB}\nabla_D N^{DB}
-\frac{1}{4}\partial_u(C_{A}^{\,\,\,\,D}\nabla^B C_{DB}) \rt] 
\end{split}
\]
since the integral of $\na_A m + \frac{1}{8}\nabla_A(C_{BE} N^{BE})$ against $\epsilon^{AB}\nabla_B \tilde{X}^k$ vanishes. Integrating by parts the first term and use $\epsilon^{AB}\epsilon_{AE}=\delta^B_E$, we obtain the desired formula. 

We now turn to the formula for $ \tilde C^k$. By \eqref{com} and \eqref{Y^A first term of duN_A},
\begin{align*}
  &\partial_u \tilde C^k \\
=& \int_{S^2} \nabla^{A}\tilde{X}^{k}\lt[ \partial_u N_A - \nabla_A m +\frac{u}{8}\nabla_A |N|^2 -\frac{1}{4}\partial_u(C_{A}^{\,\,\,\,D}\nabla^B C_{DB}) - \frac{1}{16}\nabla_{A}\partial_u(C_{DE}C^{DE}) \rt] \\
=& \int_{S^2} \nabla^{A}\tilde{X}^{k} \Big[ \frac{u}{8}\nabla_A |N|^2 + \frac{1}{8}\nabla_A(C_{BE} N^{BE})    +\frac{1}{2} C_{AB}\nabla_D N^{DB}  \\
&\qquad \qquad \quad-\frac{1}{4}\partial_u(C_{A}^{\,\,\,\,D}\nabla^B C_{DB}) - \frac{1}{16}\nabla_{A}\partial_u(C_{DE}C^{DE}) \Big]
\end{align*}
since the integral of $\frac{1}{8} \epsilon_{AB} \nabla^B (\epsilon^{PQ} C_P^{\,\,\, E}N_{EQ} )$ against $\nabla^A \tilde{X}^k$ vanishes. We arrive at the desired formula since
$ \partial_u(C_{DE}C^{DE})=2 (C_{BE} N^{BE}) $. 
\end{proof}

\section{Evolution formulae in terms of potentials}
In this section, we rewrite the evolution formulae in terms of the potentials of the shear and the news tensor.

\subsection{Energy and linear momentum}
First we recall the formulae for the energy and linear momentum.
\begin{proposition} \label{energy_momenum_evol_potential}Suppose $C_{AB}$ and $N_{AB}$ are given as in \eqref{CAB} and \eqref{NAB}, we have
\[\begin{split} \partial_u E&=-\frac{1}{8} \int_{S^2} [n \Delta (\Delta+2)n +\underline{n} \Delta (\Delta+2) \underline{n}]\ \\
\partial_u P^k&= -\frac{1}{8} \int_{S^2} \tilde{X}^k [((\Delta+2)n)^2 +((\Delta+2) \underline{n})^2- 4\epsilon^{AB} \nabla_A  n \nabla_B (\Delta+2)\underline{n}] .\\\end{split}\]

\end{proposition}

\begin{proof} These follow from Lemma \ref{news_square_integral}.

\end{proof}

\subsection{Proof of Theorem \ref{evolution2_bracket}}

We first prove the following Proposition: 
\begin{proposition}\label{evolution in terms of potentials}
The evolution formulae of the  conserved quantities can be written as 
\[
\begin{split}
\partial_{u}\tilde J^{k} = &
\frac{1}{8} \int_{S^2} \tilde{X}^k \epsilon^{AB }[\nabla_{A}{c} \nabla_{B}\Delta(\Delta + 2) {n}+ \nabla_{A}\underline{c} \nabla_{B}\Delta(\Delta + 2)\underline{n}]\\
\partial_u \tilde C^{k} =& 
 \frac{1}{8} \int_{S^2} u \tilde{X}^k [((\Delta+2)n)^2 +((\Delta+2) \underline{n})^2- 4\epsilon^{AB} \nabla_A  n \nabla_B (\Delta+2)\underline{n}] \\
&+\frac{1}{16}\int_{S^2}  [\tilde{X}^{k}(\Delta(\Delta + 2)\underline{c}(\Delta + 2)\underline{n} -\Delta(\Delta + 2)\underline{n}(\Delta + 2)\underline{c})]\\
&+ \frac{1}{16}\int_{S^2}  [\tilde{X}^{k}(\Delta(\Delta + 2)c(\Delta + 2)n - \Delta(\Delta + 2)n(\Delta + 2)c].
\end{split}
\]
\end{proposition}

\begin{proof}
We write \[4\partial_u \tilde J^k = \int_{S^2} - \tilde X^k \epsilon^{AB} C_B^{\;\;D} N_{DA} +\int_{S^2} Y_{k}^A   (C_{AB}\nabla_{D}N^{BD} -N_{AB}\nabla_{D}C^{BD})=(1)+(2)\] and compute $(1)$ and $(2)$ separately. 

Note that $(1) = -\int_{S^2} \tilde{X}^k \varepsilon_2(C_{AB}) N^{AB}$ and recall that $\varepsilon_2(C_{AB})$ has potentials $-\cb$ and $c$. Applying \eqref{polarize}, we get  
\begin{align*}
(1) =& -\frac{1}{2} \int_{S^2} \tilde{X}^k [-(\Delta+2) \underline{c} (\Delta+2)n +(\Delta+2)c (\Delta+2) \underline{n}\\
     &+  \int_{S^2} \epsilon^{AB} (\nabla_A  \underline{c} \nabla_B (\Delta+2)\underline{n}-\nabla_A  n \nabla_B (\Delta+2)c) ] \\
=& -\frac{1}{2} \int_{S^2} \tilde{X}^k [-(\Delta+2) \underline{c} (\Delta+2)n +(\Delta+2)c (\Delta+2) \underline{n}] \\
& - \int_{S^2} \tilde{X}^k   \epsilon^{AB} [\nabla_A  \underline{c} \nabla_B (\Delta+2)\underline{n}-\nabla_A  n \nabla_B (\Delta+2)c]\\
=& -\frac{1}{2} \int_{S^2} \tilde{X}^k [-(\Delta+2) \underline{c} (\Delta+2)n +(\Delta+2)c (\Delta+2) \underline{n}] \\
& - \int_{S^2} \tilde{X}^k   \epsilon^{AB} [\nabla_A {c} \nabla_B (\Delta+2){n}
+\nabla_A  \underline{c} \nabla_B (\Delta+2)\underline{n}]
\end{align*}
where we used \eqref{ibp1} in the last equality. 

Applying \eqref{YCN_bracket} to $Y^A=Y^A_k$, we have 
\begin{align*}
(2) &= \frac{1}{2}\int_{S^2} \tilde{X}^k   \epsilon^{AB} [\nabla_A (\Delta+2)c \nabla_B (\Delta+2)n+\nabla_A (\Delta+2)\underline{c} \nabla_B (\Delta+2)\underline{n}]\\
&\quad +\frac{1}{2}\int_{S^2}  \tilde{X}^k [(\Delta+2)c(\Delta+2)\underline{n}-(\Delta+2)\underline{c} (\Delta+2)n]
\end{align*}

Therefore, 
\begin{align*}
(1) + (2) &= - \int_{S^2} \tilde{X}^k   \epsilon^{AB} [\nabla_A {c} \nabla_B (\Delta+2){n}
+\nabla_A  \underline{c} \nabla_B (\Delta+2)\underline{n}]\\
&\quad +\frac{1}{2}\int_{S^2} \tilde{X}^k   \epsilon^{AB} [\nabla_A (\Delta+2)c \nabla_B (\Delta+2)n+\nabla_A (\Delta+2)\underline{c} \nabla_B (\Delta+2)\underline{n}]\\
&= \frac{1}{2} \int_{S^2} \tilde X^k \big[ \epsilon^{AB} \na_A \Delta c \na_B (\Delta + 2)n + \na_A \Delta \cb \na_B (\Delta + 2)\nb \big], 
\end{align*}
and the desired formula follows by \eqref{ibp1}.

As for the evolution of the center of mass, we apply \eqref{YCN} and note that
\[\int_{S^2} \nabla^{A}\tilde{X}^{k} \epsilon_A^{\,\,\,\,D} \nabla_D [(\Delta+2)c(\Delta+2)\underline{n}+(\Delta+2)\underline{c} (\Delta+2)n]=0.\]

Therefore,

\[
\begin{split}
\partial_u \tilde C^{k} 
=&   \frac{1}{8} \int_{S^2} u \tilde{X}^k [((\Delta+2)n)^2 +((\Delta+2) \underline{n})^2- 4\epsilon^{AB} \nabla_A  n \nabla_B (\Delta+2)\underline{n}] \\
&-\frac{1}{16}\int_{S^2} [\nabla^{A}\tilde{X}^{k}(\nabla_{A}(\Delta + 2)\underline{c}(\Delta + 2)\underline{n} -\nabla_{A}(\Delta + 2)\underline{n}(\Delta + 2)\underline{c})]\\
&-\frac{1}{16}\int_{S^2} [\nabla^{A}\tilde{X}^{k}(\nabla_{A}(\Delta + 2)c(\Delta + 2)n - \nabla_{A}(\Delta + 2)n(\Delta + 2)c]\\
=&  \frac{1}{8} \int_{S^2} u\tilde{X}^k [((\Delta+2)n)^2 +((\Delta+2) \underline{n})^2- 4\epsilon^{AB} \nabla_A  n \nabla_B (\Delta+2)\underline{n}] \\
&+ \frac{1}{16}\int_{S^2} \tilde{X}^{k}[\Delta(\Delta + 2)\underline{c}(\Delta + 2)\underline{n} -\Delta(\Delta + 2)\underline{n}(\Delta + 2)\underline{c}]\\
&+\frac{1}{16}\int_{S^2} \tilde{X}^{k}[\Delta(\Delta + 2)c(\Delta + 2)n - \Delta(\Delta + 2)n(\Delta + 2)c]\\
\end{split}
\]
\end{proof}

To obtain the formulae given in Theorem \ref{evolution2_bracket}, we rewrite the above formulae in terms of bracket operators on $S^2$.

\begin{definition}\label{bracket} For two smooth functions $u$ and $v$ on $S^2$, denote
\begin{equation}\label{bracket1} [u, v]_1=\epsilon^{AB}\nabla_A u\nabla_B v\end{equation} and
\begin{equation}\label{bracket2} [u, v]_2=\frac{1}{2} ((\Delta u)v-(\Delta v) u).\end{equation} 
\end{definition}
In view of Definition \ref{bracket}, we can write 
\[\partial_u \tilde J^k=\frac{1}{8} \int_{S^2} \tilde{X}^k ( [ c, \Delta(\Delta+2) n]_1+[ \underline{c}, \Delta(\Delta+2) \underline{n}]_1)\] and similarly for the center of mass. This proves Theorem \ref{evolution2_bracket}.
%\[\partial_u \tilde C^k=\frac{1}{8}\int \tilde{X}^k(16m+ [ (\Delta+2) c, (\Delta+2) n]_2+[ (\Delta+2) \underline{c}, (\Delta+2) \underline{n}]_2),\] and 

%\section{Consequences of the evolution formulae}
\section{Supertranslation invariance of the total flux}

\subsection{Total flux of classical conserved quantities}
We study the effect of supertranslation on the total flux of conserved quantities along null infinity or, equivalently, the difference of conserved quantities at timelike infinity and spatial infinity. As in the previous section, suppose $I=(-\infty,\infty)$ and $\mathscr{I}^+$ is complete extending from spatial infinity ($u=-\infty$) to timelike infinity ($u=+\infty$). A supertranslation is a change of coordinates $(\bar{u}, \bar x^A)\rightarrow (u, x^A)$ such that $u = \bar u + f(x), x^A = \bar x^A$ on $\mathscr{I}^+$. Let $m$, $C_{AB}$, and $N_{AB}$ denote the mass aspect, the shear, and the news, respectively, in the $(u,  x^A)$ coordinate system. Since the spherical coordinate is unchanged, we use $x$ to denote either $x^A$ or $\bar x^A$ throughout this section. It is well-known (see \cite[(C.117) and (C.119)]{CJK} for example) that the shear $\bar{C}_{AB}(\bar u,x)$, and the news $ \bar{N}_{AB}(\bar u,x)$  in the $(\bar u, x)$ coordinate system are given by 
\begin{align}
\bar{C}_{AB}(\bar u,x) &= C_{AB}(\bar u+ f(x),x) - 2 \na_A\na_B f + \Delta f \sigma_{AB} \label{supertranslation shear}\\
\bar{N}_{AB}(\bar u,x) &= N_{AB}(\bar u+f(x),x) \label{supertranslation news}
\end{align}

We assume that there exists a constant $\varepsilon>0$ such that
\begin{align}\label{news_decay}
N_{AB}( u,x) = O(|u|^{-1-\varepsilon}) \mbox{ as } u \rw \pm\infty.
\end{align}
Note that the limits of the shear tensor exist 
\[ \lim_{ u \rw \pm\infty} C_{AB} (u,x) = C_{AB}(\pm)\]
as a result of \eqref{news_decay}. 

Similarly, \eqref{news_decay} implies that the limits of the angular momentum exist
\[ \lim_{u\rightarrow \pm \infty} \tilde J^k (u) =\tilde J^k(\pm).  \]
Denote the corresponding quantities after supertranslation by $\tilde J^k_f(\pm)$.

By \eqref{dujk}, the total angular momentum flux is
\begin{align}\label{total_angmom_flux}
\begin{split}
&\tilde J^k(+) - \tilde J^k(-) \\
&= \frac{1}{4}\int_{-\infty}^{+\infty} \int_{S^2} \lt[  Y^A \lt( C_{AB}\na_D N^{BD} - N_{AB} \na_D C^{BD} \rt) + \tilde X^k \epsilon^{AB} C_A^{\;\;D} N_{DB} \rt] (u, x) dS^2 du\\
&= \frac{1}{4}\int_{-\infty}^{+\infty} \int_{S^2} \lt[  -\na_D Y^A C_{AB}N^{BD} + Y^A \lt( -\na_D C_{AB} N^{BD} - N_{AB} \na_D C^{BD} \rt)  \rt] (u, x) dS^2 du\\
&\quad + \frac{1}{4}\int_{-\infty}^{+\infty}\int_{S^2} \tilde X^k \epsilon^{AB} (C_A^{\;\;D} N_{DB})(u,x)\,dS^2du 
\end{split}
\end{align}
in the $(u, x)$ coordinates and 
\begin{equation}\label{Jf}\begin{split}
&\tilde J^k_f(+) - \tilde J^k_f(-) \\
&= \frac{1}{4}\int_{-\infty}^{+\infty} \int_{S^2} \lt[  -\na_D Y^A \bar C_{AB} \bar N^{BD} + Y^A \lt( -\na_D \bar C_{AB} \bar N^{BD} - \bar N_{AB} \na_D \bar C^{BD} \rt)  \rt] (\bar u, x) dS^2 d\bar u\\
&\quad + \frac{1}{4}\int_{-\infty}^{+\infty} \int_{S^2}\tilde X^k \epsilon^{AB} (\bar C_A^{\;\;D} \bar N_{DB})(\bar u,x)\,dS^2d\bar u\end{split}
\end{equation}
in the $(\bar{u}, x)$ coordinates.

Applying the chain rule on \eqref{supertranslation shear} yields
\begin{align*}
\na_D \bar C_{AB} (\bar u,x)&= N_{AB}(\bar u + f,x)\na_D f + (\na_D C_{AB})(\bar u+f,x) -\na_D F_{AB},\\ 
\na_D \bar{C}^{BD}(\bar u,x) &= N^{BD}(\bar u+f,x) \na_D f + (\na_D C^{BD}) ( \bar u+f ,x) -\na^B (\Delta+2)f.
\end{align*}
To simplify notation, we introduce the $u$ independent symmetric traceless 2-tensor \[ F_{AB} = 2\na_A\na_B f - \Delta f \sigma_{AB} \] and thus $\nabla_D F^{BD}=\nabla^B (\Delta+2)f$.
 
Equation \eqref{Jf} can be rewritten as
\begin{equation}\label{Jf_temp}\begin{split}
&\tilde J^k_f(+) - \tilde J^k_f(-)\\
=& \frac{1}{4}\int_{-\infty}^{\infty}\int_{S^2} \lt[ -\na_D Y^A (C_{AB}-F_{AB})N^{BD} + Y^A \omega_A  \rt]  (\bar u+f, x) dS^2 d\bar u \\
& + \frac{1}{4}\int_{-\infty}^{+\infty}\int_{S^2} \lt[ \tilde X^k \epsilon^{AB} \lt( C_A^{\;\;D} -F_A^D \rt)N_{DB} \rt](\bar u+f,x) \,dS^2 d\bar u
\end{split}\end{equation}
where
\begin{align*}
\omega_A (u, x)&= \big( -N_{AB}(u,x)\na_D f -\na_D C_{AB} + \na_D F_{AB} \big) N^{BD}(u,x)\\
&\quad - N_{AB}( u,x) \lt( N^{BD}( u,x)\na_D f + \na_D C^{BD}( u,x) - \na_D F^{BD}(x) \rt).
\end{align*}
Note that the integrand is evaluated at $(\bar u+f, x)$ in equation \eqref{Jf_temp}, to which the change of variable will be applied.

By the decaying assumption of the news \eqref{news_decay}, we can apply change of variable $u=\bar u+f$ to \eqref{Jf_temp} and rewrite it as

\begin{align}\label{Jf_temp1}
\begin{split}
&\tilde J^k_f(+) -\tilde J^k_f(-)\\
&= 
\frac{1}{4}\int_{-\infty}^{+\infty}\int_{S^2} \lt[ -\na_D Y^A (C_{AB}-F_{AB})N^{BD} + Y^A \omega_A + \tilde X^k \epsilon^{AB} \lt( C_A^{\;\;D} -F_A^D \rt)N_{DB}  \rt] (u, x)dS^2 du
\end{split}\end{align}

Combining \eqref{total_angmom_flux} and   \eqref{Jf_temp1}, we obtain 
\begin{align*}
&\lt( \tilde J^k_f(+) - \tilde J^k_f(-) \rt) - \lt( \tilde J^k(+) - \tilde J^k(-) \rt)\\
&= \frac{1}{4} \int_{-\infty}^\infty \int_{S^2} -Y^A |N|^2 \na_A f dS^2  du\\
&\quad + \frac{1}{4} \int_{-\infty}^\infty \int_{S^2} \lt[ -Y^A F_{AB}\na_D N^{BD} + Y^A N_{AB}\na_D F^{BD} - \tilde X^k \epsilon^{AB} F_A^D N_{DB}\rt] dS^2 du
\end{align*}
where we used the identity $2N_{AB}N^{BD} = |N|^2 \delta_A^D$. 

Observe that the second integral is of the same form as $\pl_u \tilde J$ given in \eqref{dujk} and one can thus simplify it as in the proof of Proposition \ref{evolution in terms of potentials} to get 
\begin{align*}
&\lt( J^k_f(+) - J^k_f(-) \rt) - \lt( J^k(+) - J^k(-) \rt)\\
=& \frac{1}{4} \int_{-\infty}^\infty \int_{S^2}  f Y^A \na_A |N|^2  dS^2 du+ \frac{1}{4}  \int_{-\infty}^\infty \int_{S^2} \tilde X^k \epsilon^{AB} \na_A n \na_B \Delta(\Delta+2) f dS^2 du
\end{align*} 
Integrating by parts, we arrive at
\begin{align}\label{diff_angmom_flux_supertranslation}
\begin{split}
&\lt( \tilde J^k_f(+) - \tilde J^k_f(-) \rt) - \lt( \tilde J^k(+) -\tilde J^k(-) \rt)\\
=& \frac{1}{4} \int_{-\infty}^{+\infty} \int_{S^2} f Y^A \na_A \big( |N|^2 - \Delta (\Delta+2)n \big) dS^2 du\\
= & \int_{S^2} -2f Y^A \na_A (m(+)-m(-)) dS^2
\end{split}
\end{align}
where 
\begin{align*}
m(\pm) = \lim_{u \rw\pm \infty} m( u,x).
\end{align*} Here we used the mass loss formula \eqref{m_aspect_evol} in the form $\pl_u m = \frac{1}{8} \Delta(\Delta + 2)n - \frac{1}{8} |N|^2$.
Note that  $m(+)-m(-)$ is of the same mode as $Y^A \na_A (m(+)-m(-))$ because $Y^A$ is a Killing field. 

In summary, we obtain a necessary and sufficient condition for the total flux of the classical angular momentum to be supertranslation invariant.

\begin{theorem}\label{supertranslation invariance of total flux}
Suppose the news tensor decays as
\[ N_{AB}(u,x) = O(|u|^{-1-\varepsilon}) \mbox{ as } u \rw \pm\infty.
\]
The total flux of the classical angular momentum $\tilde J^k$ is  supertranslation invariant if and only if 
\[
m(+) - m(-)
\]
(as a function on $S^2$) is supported in the $l \le 1$ modes.

Moreover, the above condition holds when the rescaled curvature components $P$ (see Definition \ref{def curv components}) at $\mathscr{I}^+$ satisfy
\begin{align}
\lim_{u \rw \infty}P -\lim_{u \rw -\infty}P
\end{align}
is supported in the $l \le 1 $ modes.
\end{theorem}

\begin{remark} Theorem \ref{supertranslation invariance of total flux} is motivated by the investigation in \cite{Christo1991}, which is built on the framework of stability of Minkowski spacetime. Indeed, equation (11) and (12) of \cite{Christo1991}
\[ Z^+ - Z^- = \na\Phi \qquad \mbox{div} \lt( \Sigma^+ - \Sigma^- \rt) = Z^+ - Z^- \]
imply $\lim_{u\rw\infty} \cb(u,x)= \lim_{u \rw -\infty} \cb(u,x)$. Using moreover (10) of \cite{Christo1991}
\[ \Delta \Phi = -2(F - \bar F), \]
we get
\[ \na_A \big( 8F - \Delta(\Delta+2)c|_{-\infty}^{+\infty}\big) =0. \]
\end{remark}

Moreover, the total flux of the classical center of mass is supertranslation invariant under the same condition
\begin{theorem} Suppose the news tensor decays as
\[ N_{AB}( u,x) = O(|u|^{-1-\varepsilon}) \mbox{ as } u \rw \pm\infty,
\] 
The total flux of the classical center of mass $\tilde C^k$ is  supertranslation invariant if and only if 
\[
m(+) - m(-)
\]
is a constant function on $S^2$.
\end{theorem}
\begin{proof} Denoting $\tilde C^k(\pm)=\lim_{u\rw\pm\infty} \tilde C^k(u)$, 
by \eqref{duck} we have
\begin{align}\label{classical com total flux}
\begin{split}
&\tilde C^k(+) -\tilde C^k(-) \\
=& \frac{1}{4} \int_{-\infty}^{+\infty}\int_{S^2} u|N|^2(u,x) \tilde X^k + \na^A \tilde X^k \lt[ C_{AB} \na_D N^{BD} - N_{AB} \na_D C^{BD} \rt] (u, x) \,dS^2du.
\end{split}
\end{align}

On the other hand, 
\begin{align*}
&\tilde C_f^k(+) - \tilde C_f^k(-) \\
= & \frac{1}{4} \int_{-\infty}^{+\infty}\int_{S^2} \bar u|\bar N|^2(\bar u,x) \tilde X^k + \na^A \tilde X^k \lt[ \bar C_{AB} \na_D \bar N^{BD} - \bar N_{AB} \na_D \bar C^{BD} \rt] (\bar u, x) \,dS^2d\bar u.
\end{align*}
Proceed in the same way as in the case of angular momentum, we obtain
\begin{align*}
&\lt( \tilde C^k_f(+) - \tilde C^k_f(-)\rt) - \lt( \tilde C^k(+) - \tilde C^k(-) \rt)\\
=& \frac{1}{4} \int_{-\infty}^{+\infty}\int_{S^2} -\tilde X^k|N|^2 f - \na^A \tilde X^k  |N|^2 \na_A f \,dS^2du\\
& + \frac{1}{4} \int_{-\infty}^{+\infty}\int_{S^2} \na^A \tilde X^k (-2\na_A\na_B f + \Delta f \sigma_{AB})\na_D N^{BD} + \na^A \tilde X^k N_{AB} \na^B (\Delta + 2 )f \,dS^2du
\end{align*}
We simplify the second integral as
\begin{align*}
&\int_{S^2} -2 \tilde X^k \na_A f \na_D N^{AD} + 2 \na^A \tilde X^k \na_A f \na_B\na_D N^{BD} - 2 \na^A \tilde X^k \na^B N_{AB} \cdot f\\
=& \int_{S^2} 2\lt( \na^A \tilde X^k \na_A f \na_B \na_D N^{BD} + \tilde X^k \na^A\na^B N_{AB} \cdot f \rt)  
\end{align*}
and the mass loss formula $\pl_u m = \frac{1}{4} \na^A\na^B N_{AB} - \frac{1}{8}|N|^2$ implies that

\begin{align*}
&\lt( \tilde C^k_f(+) - \tilde C^k_f(-)\rt) - \lt( \tilde C^k(+) -\tilde C^k(-) \rt) \\
= & \int_{S^2} 2 \tilde X^k f \big( m(+)-m(-) \big) + 2\na^A \tilde X^k \big( m(+)-m(-) \big) \na_A f\\
= &\int_{S^2} \lt( 6\tilde X^k \big( m(+)-m(-) \big) - 2 \na^A \tilde X^k \na_A \big(m(+)-m(-) \big) \rt) f.
\end{align*}
 
Hence, $\tilde C^k(+) -\tilde C^k(-)$ is invariant under arbitrary supertranslation if and only if
$6\tilde X^k \big( m(+)-m(-) \big) - 2 \na^A \tilde X^k \na_A \big( m(+)-m(-) \big) $ is supported in the $l \le 1$ modes.

Multiplying the expression by $\tilde X^k$ and summing over $k=1,2,3$, we get 
$m(+)-m(-)$ is supported in the $l \le 2$ modes. However, a direct computations shows that if $m(+)-m(-)$ contains a $l=2$ mode, then 
$6\tilde X^k \big( m(+)-m(-) \big) - 2 \na^A \tilde X^k \na_A \big( m(+)-m(-) \big) $ contains a $l=3$ mode. Simiarly, if $m(+)-m(-)$ contains a $l=1$ mode, then 
$6\tilde X^k \big( m(+)-m(-) \big) - 2 \na^A \tilde X^k \na_A \big( m(+)-m(-) \big) $ contains a $l=2$ mode. Thus, $m(+)-m(-)$  is constant if and only if $\tilde C^k(+) -\tilde C^k(-)$ is invariant under arbitrary supertranslation
\end{proof}

\subsection{Total flux of the CWY conserved quantities}
In this subsection, we show that the total flux of the CWY angular momentum and center of mass is supertranslation invariant. We decompose $f$ into its modes: 
\[ f = \alpha_0 + \alpha_i \tilde X^i + f_{l \ge 2}\]
and let $J^k (\pm)$ be the limits of the CWY angular momentum in the $u$ coordinate and  $J_f^k (\pm)$ be the limits of the CWY angular momentum in the $\bar u$ coordinate. We have
\begin{theorem}\label{total_flux_CWY_angular}
Suppose the news tensor decays as
\[ N_{AB}( u,x) = O(|u|^{-1-\varepsilon}) \mbox{ as } u \rw \pm\infty.
\]
Then the total flux of $ J^k$ is supertranslation invariant. Namely,
\begin{align*}
\lt( J_f^k (+)-  J_f^k (-)   \rt)      -\lt( J^k (+)- J^k (-)\rt) =  \alpha_i \varepsilon^{ik}_{\;\;\;j} (P^j(+)-P^j(-)).
\end{align*}

\end{theorem}
\begin{proof}
Note that  
\begin{equation}\label{CWY_AM}  J^k = \tilde J^k - \int_{S^2} Y^A c \na_A m. \end{equation}

The assumption \eqref{news_decay} on the decay of news tensor implies that the limit of mass aspect function is invariant of supertranslation
\begin{equation}\label{mass_diff} \lim_{\bar u \rw \pm\infty} \bar {m}(\bar u,x) =  \lim_{u \rw \pm\infty}  m(u,x) \text{ or } \bar{m}(\pm)=m(\pm).\end{equation}
Moreover, we have 
\begin{equation}\label{shear_limit}
\lim_{\bar u\rw\pm\infty} \bar{C}_{AB}(\bar u,x) = \lim_{u\rw \pm\infty} C_{AB}(u,x) -2\na_A\na_B f + \Delta f \sigma_{AB}. 
\end{equation}
If we denote the closed potential of $\lim_{\bar u \rw +\infty} \bar{C}_{AB}$ and  $\lim_{\bar u \rw +\infty} {C}_{AB}$ by $\bar{c}(+)$ and $c(+)$
respectively, we have
\begin{equation}\label{potential_diff} \bar{c}(+) = c(+) - 2f_{\ell\ge 2} \end{equation}
as functions on $S^2$. Evaluating the definition of the CWY angular momentum \eqref{CWY_AM} at $+\infty$ gives

\[  J^k (+) =\tilde J^k (+) - \int_{S^2} Y^A c (+) \na_A m(+),\] and 
\[ J_f^k (+) = \tilde J_f^k (+) - \int_{S^2} Y^A \bar c (+) \na_A \bar m(+). \]

Taking the difference and applying \eqref{mass_diff} and  \eqref{potential_diff}, we derive
\[  J_f^k (+)-J^k (+)  = \tilde J_f^k (+) - \tilde J^k (+) +2\int_{S^2} f_{\ell\ge2} Y^A \na_A m(+). \]
We derive a similar equation at $-\infty$ and thus
\begin{align*}
&\lt(J_f^k (+)- J_f^k (-)   \rt)      -\lt( J^k (+)-J^k (-)\rt) \\
=& \lt( \tilde J_f^k (+)- \tilde J_f^k (-)\rt) - \lt(\tilde J^k (+)-\tilde J^k (-)\rt) +2\int_{S^2} f_{\ell\ge2} Y^A \na_A (m(+)-m(-))\\
=& -2 \int_{S^2} f_{\ell\le1} Y^A \na_A (m(+)-m(-))\end{align*} 
by \eqref{diff_angmom_flux_supertranslation}. It follows that 
\begin{align*}
\lt( J_f^k (+)-  J_f^k (-)   \rt)      -\lt( J^k (+)- J^k (-)\rt) =  \alpha_i \varepsilon^{ik}_{\;\;\;j} (P^j(+)-P^j(-)).
\end{align*}
\end{proof}

Let $C^k (\pm)$ be the limits of the CWY center of mass in the $u$ coordinate and  $C_f^k (\pm)$ be the limits of the CWY center of mass in the $\bar u$ coordinate. 

\begin{theorem} \label{total_flux_CWY_COM}
Suppose the news tensor decays as
\[ N_{AB}( u,x) = O(|u|^{-1-\varepsilon}) \mbox{ as } u \rw \pm\infty,
\]
then the total flux of $ C^k$ is supertranslation invariant. Namely, 
\[ \lt(  C_f^k (+) - C^k_f (-) \rt) - \lt( C^k (+) - C^k(-) \rt) = \alpha_0 \lt( P^k(+) - P^k(-) \rt) + \alpha_k \lt( E(+) - E(-)\rt).
\]
\end{theorem}
\begin{proof} We write
\begin{equation}\label{CWY_COM} C^k=\tilde C^k- \int_{S^2} c \na^A\tilde{X}^k \na_A m+3\int_{S^2} c\tilde{X}^k m+\Xi(\underline{c}, m),\end{equation}
where $\Xi$ is an integral over $S^2$ that involves only $\underline{c}$ and $m$. Since the last three integrals have limits at $u = \pm\infty$, the mass loss formula now implies
\begin{align*}
& C^k(+) -  C^k(-) \\
=& \frac{1}{4} \int_{-\infty}^{+\infty}\int_{S^2} u|N|^2(u,x) \tilde X^k + \na^A \tilde X^k \lt[ C_{AB} \na_D N^{BD} - N_{AB} \na_D C^{BD} \rt] (u, x) \,dS^2du\\
& - \int_{S^2} c(+) \na^A\tilde{X}^k \na_A m(+) + \int_{S^2} c(-) \na^A\tilde{X}^k \na_A m(-)\\
& + 3\int_{S^2} c(+)\tilde X^k m(+) - 3\int_{S^2} c(-)\tilde X^k m(-)\\
& + \Xi(\cb(+),m(+)) - \Xi(\cb(-),m(-)).
\end{align*}
By \eqref{shear_limit}, $\cb(\pm)$ is invariant under supertranslation. We apply \eqref{mass_diff} and \eqref{potential_diff} to get
\begin{align*}
& \lt( C_f^k (+) - C^k_f (-) \rt) - \lt(  C^k (+) -  C^k(-) \rt) \\
=& \lt( \tilde C^k_f(+) - \tilde C^k_f(-)\rt) - \lt( \tilde  C^k(+) -   \tilde C^k(-) \rt) \\
& +2 \int_{S^2} f_{\ell\ge2} \na^A \tilde{X}^k \na_A \big( m(+) - m(-) \big) - 6\int_{S^2} f_{\ell\ge2} \tilde{X}^k \big( m(+) - m(-) \big)\\
=& -2 \int_{S^2} f_{\ell\le1} \na^A \tilde X^k \na_A \big( m(+) - m(-) \big) + 6 \int_{S^2} f_{\ell\le1} \tilde X^k \big( m(+) - m(-) \big).
\end{align*}
We obtain
\begin{align*}
\lt(  C_f^k (+) - C^k_f (-) \rt) - \lt( C^k (+) - C^k(-) \rt) &= 2 \int_{S^2} ( \alpha_0 \tilde X^k + \alpha_k ) ( m(+) - m(-) )\\
&= \alpha_0 \lt( P^k(+) - P^k(-) \rt) + \alpha_k \lt( E(+) - E(-)\rt).
\end{align*}
\end{proof}

\section{Spacetime with zero news}
In this section, we consider a non-radiative spacetime in the sense that the news vanishes. This includes all model spacetimes such as Minkowski and Kerr. First, we show that the CWY angular momentum and center of mass are constant.
\begin{lemma}

Suppose the news $N_{AB}(u, x)\equiv 0$ in a Bondi-Sachs coordinate system $(u, x)$, then the CWY angular momentum $J^k(u)$ and CWY center of mass $C^k(u)$ are constant, i.e. independent of the retarded
time $u$. 
\end{lemma}

\begin{proof} The assumption implies $\pl_u m(u, x)=0, \pl_u C_{AB}(u, x)=0$ and thus 

\[ m(u, x)\equiv \mathring{m}(x), C_{AB}(u, x)\equiv \mathring{C}_{AB}(x)\] and both potentials $c$ and $\underline{c}$ are independent of $u$ as well. 

We recall the definition of CWY angular momentum 
\begin{align}
 J^{k} (u)= \int_{S^2}&  Y^A \lt( N_{A} -\frac{1}{4}C_{AB}\nabla_{D}C^{DB} - c\nabla_{A}m   \rt).
\end{align} 

Since $c$ and $m$ are both independent of $u$, our previous calculation shows  \[\begin{split}&\partial_u  \int_{S^2}  Y^A \lt( N_{A} -\frac{1}{4}C_{AB}\nabla_{D}C^{DB}\rt)\\ =&\frac{1}{4}\int_{S^2} \lt[ Y^A  (C_{AB}\nabla_{D}N^{BD} -N_{AB}\nabla_{D}C^{BD}) +  \tilde{X}^k  \epsilon^{AB} (C_{A}^{\,\,\,\,D} N_{DB}) \rt],\end{split}\] the conclusion follows. 

On the other hand, the CWY center of mass $C^k$ is given by \begin{align}
\begin{split}
& \int_{S^2} \nabla^{A}\tilde{X}^{k} \lt(N_A - \frac{1}{4} C_{AB} \na_D C^{DB} - \frac{1}{16} \na_A \lt( C^{DE}C_{DE}\rt)\rt)  - \na^A \tilde{X}^k (c+u) \na_A m \\
&+\int_{S^2} \lt( 3 \tilde X^k cm+ 2 \nabla^{A}\tilde{X}^{k} \epsilon_{AB} (\na^B \cb) m      - \frac{1}{16} \tilde X^k \na_A (\Delta+2)\cb \na^A (\Delta+2)\cb\rt)\\
\end{split}
\end{align}

Since all $m$, $c$, and $\underline{c}$ are independent of $u$, 
 \begin{align}
\begin{split}  &\pl_u C^k \\
= &\pl_u \int_{S^2} \nabla^{A}\tilde{X}^{k} \lt(N_A - \frac{1}{4} C_{AB} \na_D C^{DB} - \frac{1}{16} \na_A \lt( C^{DE}C_{DE}\rt)\rt)-\int_{S^2}\na^A \tilde{X}^k  \na_A m\end{split}\end{align}
Our previous calculation shows that the first term on the right hand side is
\[\int_{S^2}  \nabla^{A}\tilde{X}^{k} \lt( \na_A m + \frac{1}{4}C_{AB}\nabla_{D}N^{BD} -\frac{1}{4}N_{AB}\nabla_{D}C^{BD} \rt),\] and the conclusion follows. 

\end{proof}

Finally, we show that in a spacetime with vanishing news tensor, the angular momentum and center of mass themselves, not just their total flux, are invariant under supertranslation.

We pin down the exact formula for the angular momentum aspect on a spacetime with vanishing news.  In this case, we have
\[\partial_u N_A  (u, x)= \na_A m (u, x) -\frac{1}{4}\nabla^B P_{BA} (u, x)\]
where \[ P_{BA} (u, x)=(\nabla_B \nabla^E C_{EA}-\nabla_A \nabla^E C_{EB})(u, x)\]

Therefore \[\partial_u N_A  (u, x)= \na_A \mathring{m} (x) -\frac{1}{4}\nabla^B \mathring{P}_{BA} (x)\] is independent of $u$. Integrating gives
\begin{equation}\label{AM_aspect1} N_A  (u, x)=  N_A  (u_0, x) +(u-u_0)       (\na_A \mathring{m} -\frac{1}{4}\nabla^B \mathring{P}_{BA} )\end{equation} for any $u$ and fixed $u_0$. 

Suppose $(\bar{u}, x)$ is another Bondi-Sachs coordinate system that is related to $(u, x)$ by a supertranslation $u=\bar{u}+f$ for $f\in C^\infty(S^2)$. 

Recall the mass aspect $\bar{m} (\bar u, x)
$, the shear $\bar{C}_{AB}(\bar u,x)$, and the news $ \bar{N}_{AB}(\bar u,x)$  in the $(\bar u, x)$ coordinate system are related to the mass aspect ${m} (u, x)
$, the shear ${C}_{AB}(u,x)$, and the news $ {N}_{AB}(u,x)$  in the $( u, x)$ coordinate system through:

\begin{equation}\label{supertranslation}\begin{split}
\bar{m}(\bar u,x) &= m(\bar u+f,x) + \frac{1}{2} (\na^B N_{BD})(\bar u+f,x) \na^D f \\
& + \frac{1}{4} (\partial_{u}N_{BD})(\bar u+f,x) \na^B f\na^D f + \frac{1}{4} N_{BD}(\bar u+f,x) \na^B\na^D f\\
\bar{C}_{AB}(\bar u,x) &= C_{AB}(\bar u+ f(x),x) - 2 \na_A\na_B f + \Delta f \sigma_{AB} \\
\bar{N}_{AB}(\bar u,x) &= N_{AB}(\bar u+f(x),x)
\end{split}\end{equation}

In particular, $\bar{N}_{AB}(\bar u,x)\equiv 0$ and $\bar{J}^k(\bar{u})$ and $\bar{C}^k(\bar{u})$ are independent of $\bar{u}$. In addition, we have 
\begin{equation}\label{transform}\begin{split}
\bar{m}(\bar u,x) &= \mathring{\bar{m}}(x)= \mathring{m}(x) \\
\bar{C}_{AB}(\bar u,x) &= \mathring{\bar{C}}_{AB}(x)=\mathring{C}_{AB}(x) - F_{AB} \\
\mathring{\bar{c}}&=\mathring{c}-2 f_{\ell\geq 2}\\
\mathring{\bar{\underline{c}}}&=\mathring{\underline{c}}
\end{split}\end{equation} where $F_{AB}=2 \na_A\na_B f -\Delta f \sigma_{AB} $.

Finally the angular momentum aspect transforms by 
\[\begin{split}\bar{N}_A(\bar{u}, x)&=N_{A}(\bar{u}+f, x)+3m(\bar{u}+f, x)\nabla_A f-\frac{3}{4}P_{BA} (\bar{u}+f, x)\na^B f\\
&=N_{A}(\bar{u}+f, x)+3\mathring{m}\nabla_A f-\frac{3}{4}\mathring{P}_{BA} \na^B f.\\
\end{split}\]
See \cite[(C.123)]{CJK}. Note that the convention of angular momentum aspect there is $-3N_A$.

Combining with \eqref{AM_aspect1} and setting $u=\bar{u}+f$, we obtain 
\begin{equation}\label{AM_aspect2} \bar{N}_A  (\bar{u}, x)=  N_A  (u_0, x) +(\bar{u}-u_0+f)       (\na_A \mathring{m} -\frac{1}{4}\nabla^B \mathring{P}_{BA} )+3\mathring{m}\nabla_A f-\frac{3}{4}\mathring{P}_{BA} \na^B f\end{equation} for any $\bar{u}$ and fixed $u_0$. 

Now fixing $\bar{u}=\bar{u}_0$, we consider the angular momentums
\[\begin{split}&\bar{J}=\bar{J}(\bar{u}_0)= \int_{S^2} Y^A \lt( \bar{N}_{A} -\frac{1}{4}\bar{C}_{AB}\nabla_{D}\bar{C}^{DB} - 
\mathring{\bar{c}}\nabla_{A}\mathring{\bar{m}}   \rt) (\bar{u}_0, x)\\
&{J}=J (u_0)= \int_{S^2} Y^A \lt( {N}_{A} -\frac{1}{4}\mathring{C}_{AB}\nabla_{D}\mathring{C}^{DB} - \mathring{{c}}\nabla_{A}\mathring{{m}}   \rt) (u_0, x)
\end{split}\]
\[\begin{split}&\bar{C}=\bar{C}(\bar{u}_0)= \int_{S^2} \na^A \tilde{X}^k \lt( \bar{N}_{A} -\frac{1}{4}\bar{C}_{AB}\nabla_{D}\bar{C}^{DB} -\frac{1}{16}\nabla_A (\bar{C}^{DE} \bar{C}_{DE}) \rt) (\bar{u}_0, x)\\
&+\int_{S^2} \lt(3 \tilde X^k \mathring{\bar{c}}\,\mathring{\bar{m}} - \na^A \tilde{X}^k (\mathring{\bar{c}}+\bar{u}_0) \na_A \mathring{\bar{m}} \rt)\\
&+ \int_{S^2} \lt(2 \nabla^{A}\tilde{X}^{k} \epsilon_{AB} (\na^B \underline{\mathring{\bar{c}}}) \mathring{\bar{m}}      - \frac{1}{16} \tilde X^k \na_A (\Delta+2) \underline{\mathring{\bar{c}}}  \na^A (\Delta+2) \underline{\mathring{\bar{c}}}      \rt)\\
&{C}=C({u}_0)= \int_{S^2} \na^A \tilde{X}^k \lt( {N}_{A} -\frac{1}{4}\mathring{C}_{AB}\nabla_{D}\mathring{C}^{DB}  -\frac{1}{16}\nabla_A (\mathring{C}^{DE} \mathring{C}_{DE})       \rt) ({u}_0, x)\\
&+\int_{S^2} \lt(3 \tilde X^k \mathring{c} \,\mathring{m} - \na^A \tilde{X}^k (\mathring{c}+u_0) \na_A \mathring{m} \rt)\\
&+ \int_{S^2} \lt(2 \nabla^{A}\tilde{X}^{k} \epsilon_{AB} (\na^B  \underline{\mathring{{c}}}      ) m      - \frac{1}{16} \tilde X^k \na_A (\Delta+2)  \underline{\mathring{{c}}}\na^A (\Delta+2) \underline{\mathring{{c}}}         \rt)\\
\end{split}\]
We prove the following theorem:
\begin{theorem} \label{invariant_CWY_vanish_news}
On a spacetime with vanishing news, the CWY angular momentum and center of mass satisfy
\begin{align}
 \bar J   - J= &-2 \int_{S^2} Y^A f_{\ell\leq 1}\na_A \mathring{m} \\
\bar C -C =&  \int_{S^2}  \lt( 6 f_{\ell \leq 1} \tilde{X}^k \mathring{m}-2f_{\ell\leq 1} \na^A \tilde{X}^k \na_A \mathring{m}  \rt)
\end{align}
\end{theorem}
\begin{proof}

Taking the difference of $\bar J $ and  $J$ and applying \eqref{transform}, we obtain 

\[\begin{split}\bar{J}-J=&\int_{S^2} Y^A \lt[\bar{N}_A  (\bar{u}_0, x)-  N_A  (u_0, x)\rt]\\
&+\frac{1}{4}\int_{S^2}  Y^A \lt[  \mathring{C}_{AB}\nabla_D F^{BD}+F_{AB} \nabla_D \mathring{C}^{BD}- F_{AB} \nabla_D {F}^{BD}    \rt]\\
&+2 \int_{S^2} Y_A f_{\ell\geq 2}\na_A \mathring{m}\end{split}\]

We observe that $\int_{S^2}  Y^A (F_{AB} \nabla_D {F}^{BD})=0$ and  compute 

\[\begin{split}&\int_{S^2} Y^A \lt[\bar{N}_A  (\bar{u}_0, x)-  N_A  (u_0, x)\rt]\\
=& \int_{S^2} Y^A\lt[ (\bar{u}-u_0+f)       (\na_A \mathring{m} -\frac{1}{4}\nabla^B \mathring{P}_{BA} )+3\mathring{m}\nabla_A f-\frac{3}{4}\mathring{P}_{BA} \na^B f \rt]\\ 
=& \int_{S^2} Y^A\lt[ f       \na_A \mathring{m} -\frac{1}{4}f \nabla^B \mathring{P}_{BA} +3\mathring{m}\nabla_A f-\frac{3}{4}\mathring{P}_{BA} \na^B f \rt]\\
=& \int_{S^2} Y^A\lt[-2 f \na_A \mathring{m} -\frac{1}{4}f \nabla^B \mathring{P}_{BA} -\frac{3}{4}\mathring{P}_{BA} \na^B f \rt],\end{split}\]
where we use $\int_{S^2} Y^A\nabla^B \mathring{P}_{BA}=0$ and $\int_{S^2} Y^A \na_A \mathring{m}=0$. Therefore,

\[
\begin{split}
\bar{J}-J=&\frac{1}{4}\int_{S^2}  Y^A \lt[  \mathring{C}_{AB}\nabla_D F^{BD}+F_{AB} \nabla_D \mathring{C}^{BD}  -f \nabla^B \mathring{P}_{BA} -3\mathring{P}_{BA} \na^B f     \rt]\\
&-2 \int_{S^2} Y_A f_{\ell\leq 1}\na_A \mathring{m}.
\end{split}
\]
By Theorem \ref{thm_int_1}, the first integral vanishes and the result follows. 

Taking the difference of $\bar C$ and $C$ and applying \eqref{transform}, we obtain 
\[\begin{split}\bar{C}-C=&\int_{S^2} \na^A \tilde{X}^k \lt[\bar{N}_A  (\bar{u}_0, x)-  N_A  (u_0, x)\rt]\\
&+\frac{1}{4}\int_{S^2} \na^A \tilde{X}^k \Big [  \mathring{C}_{AB}\nabla_D F^{BD}+F_{AB} \nabla_D \mathring{C}^{BD}-F_{AB}\nabla_D F^{DB}\\
&  \qquad \qquad \qquad \qquad+\frac{1}{2} \nabla_A (C_{BD} F^{BD})-\frac{1}{4} \na_A( F_{BD}F^{BD} )    \Big]\\
&+ \int_{S^2}  \lt( -6 f_{\ell \geq 2} \tilde{X}^k \mathring{m}+( 2f_{\ell\geq 2}-\bar{u}_0+u_0) \na^A \tilde{X}^k \na_A \mathring{m}  \rt)\end{split}\]

We observe that \[ \int_{S^2} \na^A \tilde{X}^k \lt[  -F_{AB}\nabla_D F^{DB}-\frac{1}{4} \na_A( F_{BD}F^{BD} ) \rt]=0\] 
and compute 
\[\begin{split}&\int_{S^2} \na^A \tilde{X}^k \lt[\bar{N}_A  (\bar{u}_0, x)-  N_A  (u_0, x)\rt]\\
=& \int_{S^2} \na^A\tilde{X}^k\lt[ (\bar{u}_0-u_0+f)       (\na_A \mathring{m} -\frac{1}{4}\nabla^B \mathring{P}_{BA} )+3\mathring{m}\nabla_A f-\frac{3}{4}\mathring{P}_{BA} \na^B f \rt]\\ 
=& \int_{S^2} \na^A \tilde{X}^k \lt[ (\bar{u}_0-u_0+f)      \na_A \mathring{m} -\frac{1}{4}f \nabla^B \mathring{P}_{BA} +3\mathring{m}\nabla_A f-\frac{3}{4}\mathring{P}_{BA} \na^B f \rt]\\
=& \int_{S^2} \na^A \tilde{X}^k\lt[(\bar{u}_0-u_0-2 f) \na_A \mathring{m} -\frac{1}{4}f \nabla^B \mathring{P}_{BA} -\frac{3}{4}\mathring{P}_{BA} \na^B f \rt]+\int_{S^2} (6 f\tilde{X}^k \mathring{m}),\end{split}\] 
where we use $\int_{S^2} \na^A\tilde{X}^k \nabla^B \mathring{P}_{BA}=0$.

Putting everything together, we arrive at 
\[\begin{split}
   &\bar{C}-C\\
=&\frac{1}{4}\int_{S^2} \na^A \tilde{X}^k \lt[  \mathring{C}_{AB}\nabla_D F^{BD}+F_{AB} \nabla_D \mathring{C}^{BD}+\frac{1}{2} \nabla_A (\mathring{C}_{BD} F^{BD})-f \nabla^B \mathring{P}_{BA} -3\mathring{P}_{BA} \na^B f   \rt]\\
&+ \int_{S^2}  \lt( 6 f_{\ell \leq 1} \tilde{X}^k \mathring{m}-2f_{\ell\leq 1} \na^A \tilde{X}^k \na_A \mathring{m}  \rt)\end{split}\]

By Theorem \ref{thm_int_2}, the first integral vanishes and the result follows. 
\end{proof}

\section{Conservation law of angular momentum and a duality paradigm for null infinity}
\subsection{Conservation law of angular momentum}\label{conservation}

In this subsection, we derive a conservation law of angular momentum at $\mathscr{I}^+$  \`a la Christodoulou \cite{Christo1991}.

Suppose $I=(-\infty, +\infty)$ and $\mathscr{I}^+$ is complete extending from spatial infinity ($u=-\infty$) to timelike infinity ($u=+\infty$).   Integrating the formula in Proposition \ref{AM_aspect_evol} from $-\infty$ to $+\infty$ and projecting onto the $\ell=1$ modes, we obtain

\begin{equation}\label{AM_conservation}\epsilon^{AE} \nabla_E N_A (+\infty)_{\ell=1}-\epsilon^{AE} \nabla_E N_A(-\infty)_{\ell=1}= G_{\ell=1},\end{equation}
where \[G=\int_{-\infty}^{+\infty} \frac{1}{8} \nabla^A \nabla_A (\epsilon^{PQ} C_P^{\,\,\, E} N_{EQ})+\frac{1}{2} \epsilon^{AE}\nabla_E (C_{AB}\nabla_D N^{DB})\]

Equation \eqref{AM_conservation} should be considered as a conservation law for angular momentum that complements the conservation law for linear momentum of
Christodoulou \cite[Equation (13)]{Christo1991}, which in our notation is 

\[ \hat m(+\infty)_{\ell=0, 1}- \hat m(-\infty)_{\ell=0, 1}= - F_{\ell=0, 1},\] where 
\begin{align}\label{F}
F=\frac{1}{8}\int_{-\infty}^\infty N_{AB}N^{AB}
\end{align} and follows from \eqref{evol_modified_mass}.

The above discussion can be carried over under the framework of stability of Minkowski spacetime, provided that we take Rizzi's definition of angular momentum \cite{Rizzi, Rizzi_thesis}. Recall from \cite{Christo1991, ChristoMG9} that two symmetric traceless 2-tensors $\Sigma$ and $\Xi$ are defined by 
\[ \lim_{C_u^+, r \rw \infty} r^2 \hat\chi = \Sigma, \quad \lim_{C_u^+, r \rw \infty} r \hat\chib = \Xi \]
with
\begin{align}\label{shear news}
\frac{\pl\Sigma}{\pl u} = -\frac{1}{2}\Xi.
\end{align}
See Definition \ref{def curv components} for the curvature components and their limits at null infinity.

Rizzi's definition of angular momentum \cite[(3)]{Rizzi} is given by (omitting the constant $\frac{1}{8\pi}$)
\begin{align}\label{Rizzi_definition}
L(\Omega_{(i)}) = \int_{S^2} \Omega^A_{(i)} \lt( I_A - \Sigma_{AB} \na_C \Sigma^{CB} \rt), \quad i=1,2,3
\end{align}
where he {\it assumes} that the curvature component $\beta$ satisfies $\lim_{r\rw\infty} r^3 \beta_A = -I_A$. Here $\Omega^A_{(i)}$ corresponds to $\epsilon^{AB} \na_B \tilde X^i$. In the appendix, we show that $I_A$ and $\Sigma_{AB}$ correspond to $N_A$ and $-\frac{1}{2}C_{AB}$ in Bondi-Sachs coordinate system. Hence Rizzi's definition coincides with \eqref{angmom}.

Using Bianchi identities, Rizzi derived the evolution formula \cite[(4)]{Rizzi}
\begin{align}
\frac{\pl L}{\pl u} &= \int_{S^2} \Omega^A \lt[ \Xi_{AB} \na_C \Sigma^{CB} + \frac{1}{2} \lt( \Sigma^C_B \na^B \Xi_{CA}  -\Sigma_{AB} \na_C \Xi^{CB} \rt) \rt]  \notag\\
&= \frac{1}{2} \int_{S^2} \Omega^A \lt( \Xi_{AB} \na_C \Sigma^{CB} - \Sigma_{AB} \na_C \Xi^{CB} \rt) + \na^A \Omega^B \Sigma_B^C \Xi_{CA}, \label{Rizzi_evolution}
\end{align}
where the second line is obtained by integrating by parts the term $\Omega^A \Sigma^C_B \na^B \Xi_{CA}$. 
 
\begin{remark}
The definition we take has the opposite sign to \cite[(3),(4)]{Rizzi}. The discrepancy comes from the fact that Kerr spacetime has angular momentum $-ma$ under our definition. 
\end{remark}

According to the main theorem of \cite{CK},
\begin{align}\label{Bbar decay}
\underline{B} = O(|u|^{-\frac{3}{2}})
\end{align}
as $|u| \rw \infty$, where $\underline{B} = \lim_{r\rw\infty} r^2 \underline{\beta}$. (see also \cite{Christo1991}, the paragraph after equation (8) where $\underline{B}$ is denoted by $B$ there)

Estimate \eqref{Bbar decay} and equation (2) of \cite{Christo1991}
\begin{align}\label{divXi=Bbar}
\mbox{div} \Xi = \underline{B}
\end{align}
imply that
\begin{align}\label{Xi decay}
\Xi = O(|u|^{-\frac{3}{2}})
\end{align}
as $|u| \rw \infty$ and 
\begin{align}\label{Sigma limit}
\Sigma \rw \Sigma^{\pm}
\end{align}
as $u \rw \pm \infty$.
 
By \eqref{Xi decay} and \eqref{Sigma limit}, $\int_{-\infty}^\infty \frac{\pl L}{\pl u} du$ is finite and furnishes the difference of the angular momenta at timelike infinity $(u \rw \infty)$ and spatial infinity $(u \rw -\infty)$.

We can write this in the spirit of \cite{Christo1991}. In general, the peeling fails and $\beta$ decays as $\beta = o(r^{-\frac{7}{2}})$. Christodoulou \cite{ChristoMG9} observed that Bianchi equation nevertheless implies that \begin{align*}
R = \lim_{C^+_u, r\rw\infty} r^4 \underline{D} \beta
\end{align*} exists. Moreover, one has
\begin{align}\label{Dbar beta}
R = \na P + * \na Q + 2 \Sigma \cdot \underline{B},
\end{align} 
where $(P,Q) = \lim_{r \rw\infty} (r^3 \rho, r^3 \sigma)$ and $\na, *, \cdot$ are taken with respect to standard metric $\sigma$ on $S^2$.

In order to exhibit a physically reasonable initial data set that has a complete Cauchy development without peeling, Christodoulou made the crucial assumption
\begin{align}\label{uR past limit}
\lim_{u \rw -\infty} uR = R^- \neq 0,
\end{align}
which we adopt here.

From \eqref{uR past limit} he derived \cite[(5)]{ChristoMG9} \begin{align*}
\beta = B_* r^{-4} \log r + B r^{-4} + o(r^{-4})
\end{align*}
uniformly in $u$ with 1-forms $B_*$ and $B$ on $S^2$ satisfying \cite[(6)]{ChristoMG9} \begin{align}
\frac{\pl B_*}{\pl u} &=0\, \\
\frac{\pl B}{\pl u} &= \frac{1}{2}R.
\end{align} Moreover, using Bianchi equation, he derived that
\begin{align*}
\lim_{C^+_u, r \rw \infty} r^4 \alpha = A_* \neq 0
\end{align*}
exists. $A_*$ is a symmetric traceless 2-tensor that is independent of $u$ and satisfies
\begin{align}\label{divA=-B}
\mbox{div} A_* = - B_*.
\end{align}

\begin{definition}
For a function $f$ on $S^2$, we denote the projection of $f$ on the sum of zeroth and first eigenspaces of $\Delta$ by $f_{[1]}$. Namely, $f_{[1]} = f_{\ell=0} + f_{\ell=1}$. For a 1-form $\omega_A = \na_A f + \epsilon_{AB} \na^B g$, we denote $\omega_{A[1]} = \na_A f_{\ell=1} + \epsilon_{AB}\na^B g_{\ell=1}$.
\end{definition}

Since the spherical tangent vectors $\pl_A$ have length $O(r)$, we have the correspondence
\begin{align}
B_A = - I_A.
\end{align} 
By \eqref{divA=-B}, $B_{*[1]}=0$ and we integrate\eqref{Dbar beta} to get
\begin{align*}
\lt( I_A(u_2,x) - I_A(u_1,x) \rt)_{[1]} = -\frac{1}{2} \int_{u_1}^{u_2} \na_A P_{\ell=1} + \epsilon_{AB} \na^B Q_{\ell=1} + ( 2 \Sigma_{AB} \underline{B}^B )_{[1]} du.
\end{align*}

By \eqref{Bbar decay} and \eqref{Sigma limit}, the last term is integrable on $(-\infty,\infty)$. For the first two terms, we observe that equations (10, 11) of \cite{ChristoMG9}
\begin{align}
\frac{\pl P}{\pl u} &= -\frac{1}{2} \mbox{div} \underline{B} + \frac{1}{2} \Sigma \cdot \frac{\pl\Xi}{\pl u} \\
\frac{\pl Q}{\pl u} &= -\frac{1}{2} \mbox{curl} \underline{B} + \frac{1}{2} \Sigma \wedge \frac{\pl\Xi}{\pl u}
\end{align}
infer that $P_{\ell=1}(u,x) = a_i \tilde X^i + O(|u|^{-\frac{3}{2}})$ and $Q_{\ell=1}(u,x) = b_i \tilde X^i + O(|u|^{-\frac{3}{2}})$ for some constants $a_i,b_i$ independent of $u$. Thanks to the main theorem of \cite{CK}, $P - P_{\ell=0}, Q - Q_{\ell=0} = O(|u|^{-\frac{1}{2}})$, we have $a_i=0, b_i=0$. Thus $P_{\ell=1}, Q_{\ell=1}$ are also integrable on $(-\infty,\infty)$.

We conclude 
$\lim_{u\rw\infty} I_{A[1]}(u,x) - \lim_{u\rw-\infty} I_{A[1]}(u,x)$ exists and is given by
\begin{align*}
-\frac{1}{2} \int_{-\infty}^\infty
\lt(  \na_A P_{\ell=1} + \epsilon_{AB} \na^B Q_{\ell=1} +  (2\Sigma_{AB} \underline{B}^B)_{[1]} \rt) du'. \end{align*}

By \eqref{shear news} and Rizzi's definition \eqref{Rizzi_definition}, we can interpret the following formula as a conservation law of angular momentum
\begin{align}
\begin{split}
&\lim_{u\rw\infty} \lt( I_A - \Sigma_{AB}\na_C \Sigma^{CB} \rt)_{[1]}- \lim_{u\rw-\infty} \lt( I_A - \Sigma_{AB}\na_C \Sigma^{CB} \rt)_{[1]} \\
&= \frac{1}{2} \int_{-\infty}^\infty  -\na_A P_{\ell=1}- \epsilon_{AB} \na^B Q_{\ell=1} +\lt(  \Xi_{AB} \na_C \Sigma^{CB} - \Sigma_{AB} \na_C \Xi^{CB} \rt)_{[1]} du.
\end{split}\end{align} 
From Proposition \ref{tensor in BS variable} and Proposition \ref{curv in BS variable}, it follows that the co-closed part of the above conservation law is equivalent to the total flux of the classical angular momentum in a Bondi-Sachs coordinate system.

\subsection{A duality paradigm for null infinity}
In this subsection, we describe a duality paradigm for null infinity which creates a pair of dual spacetimes with the same classical conserved quantities.
\begin{corollary}[Corollary \ref{duality}]
Given a set of null infinity data $(m, N_A, C_{AB}, N_{AB})$ defined on $[u_1,u_2] \times S^2$,  there exists a dual set of null infinity data  $({m}^{*}, {N}^{*}_A, {C}^{*}_{AB}, {N}^{*}_{AB})$ that has the same (classical) energy, linear momentum, angular momentum, and center of mass.
\end{corollary}

\begin{proof} 
Define $ C^{*}_{AB} = \varepsilon_2(C_{AB})$ on $[u_1,u_2] \times S^2$. Then $ N^{*}_{AB} = \pl_u  C^{*}_{AB} = \varepsilon_2(N_{AB})$. Define $ m^{*}(u,x)$ by the differential equation
\begin{align*}
\begin{cases}
 m^{*}(u_1,x) = m(u_1,x)\\
\pl_u m^{*} = \frac{1}{4} \na^A \na^B  N^{*}_{AB} - N^{*}_{AB}  {N^{*}}^{AB}
\end{cases} 
\end{align*}
and then define $ N^{*}_A$ by the differential equation
\begin{align*}
\begin{cases}
 N^{*}_A(u_1,x) = N_A(u_1,x)\\
\pl_u  N^{*}_A = \na_A  m^{*} -\frac{1}{4}\nabla^D(\nabla_D \nabla^E  C^{*}_{EA}-\nabla_A \nabla^E  C^{*}_{ED}) \\
\qquad\qquad +\frac{1}{4}\nabla_A( C^{*}_{BE}  {N^{*}}^{BE})-\frac{1}{4}\nabla_B ( {C^{*}}^{BD} N^{*}_{DA})+\frac{1}{2}  C^{*}_{AB}\nabla_D  {N^{*}}^{DB}.
\end{cases}
\end{align*}
For this subsection alone, we denote the classical conserved quantities of the infinity data $(m, N_A, C_{AB}, N_{AB})$ by $E,P^k,J^k,C^k$ and 
denote the classical conserved quantities of the data $(m^{*}, N^{*}_A, C^{*}_{AB}, N^{*}_{AB})$ by $E^{*}, {P^{*}}^k, {J^{*}}^k, {C^{*}}^k$, we have
\[ E^{*}(u_1) = E(u_1),  {P^{*}}^k(u_1) = P^k(u_1) \]
and since ${C^{*}_A}^{\;\;D} \na^B  C^{*}_{BD} = C_A^{\;\;D} \na^B C_{BD}, C^{*}_{DE}  {C^{*}}^{DE} = C_{DE} C^{DE}$,
\[ {J^{*}}^k(u_1) = J^k(u_1), {C^{*}}^k(u_1) = C^k(u_1). \]
 
It remains to show that the evolutions of the conserved quantities are identical. Recall that the potentials of $C^{*}_{AB}$ and $ N^{*}_{AB}$ are given by $(-\cb, c)$ and $(-\nb, n)$. We observe that replacing $(c,\cb,n,\nb)$ by $(-\cb, c, -\nb, n)$ does not change the following expressions 
\begin{align*}
\partial_u E &=-\frac{1}{8} \int_{S^2} [n \Delta (\Delta+2)n +\underline{n} \Delta (\Delta+2) \underline{n}], \\
\pl_u P^k &= -\frac{1}{8} \int_{S^2} \tilde{X}^k [((\Delta+2)n)^2 +((\Delta+2) \underline{n})^2- 4\epsilon^{AB} \nabla_A  n \nabla_B (\Delta+2)\underline{n}], \\
\partial_u J^k &= \frac{1}{8} \int_{S^2} \tilde{X}^k \epsilon^{AB }[\nabla_{A}{c} \nabla_{B}\Delta(\Delta + 2) {n}+ \nabla_{A}\underline{c} \nabla_{B}\Delta(\Delta + 2)\underline{n}]\\
\partial_u  C^{k} &= 
 \frac{1}{8} \int_{S^2} \tilde{X}^k [((\Delta+2)n)^2 +((\Delta+2) \underline{n})^2- 4\epsilon^{AB} \nabla_A  n \nabla_B (\Delta+2)\underline{n}] \\
&\quad +\frac{1}{16}\int_{S^2}[\tilde{X}^{k}(\Delta(\Delta + 2)\underline{c}(\Delta + 2)\underline{n} -\Delta(\Delta + 2)\underline{n}(\Delta + 2)\underline{c})]\\
&\quad + \frac{1}{16}\int_{S^2}[\tilde{X}^{k}(\Delta(\Delta + 2)c(\Delta + 2)n - \Delta(\Delta + 2)n(\Delta + 2)c]
\end{align*}
This finishes the proof.
\end{proof}

\section{The case of quadrupole moments}
In this section, we consider the case of generalized quadrupole moments. Namely, all $c, \underline{c}, n, \underline{n}$ are $(-6)$ eigenfunctions (or $\ell=2$ spherical harmonics). Therefore, ${c}=\sum c_{ij}(u) \tilde{X}^i \tilde{X}^j,  \underline{c}=\sum \underline{c}_{ij}(u) \tilde{X}^i \tilde{X}^j$ ${n}=\sum n_{ij}(u) \tilde{X}^i \tilde{X}^j,  \underline{n}=\sum \underline{n}_{ij}(u) \tilde{X}^i \tilde{X}^j$ with $\partial_u c_{ij}=n_{ij}$ and $\partial_u \underline{c}_{ij}=\underline{n}_{ij}$. 
\subsection{Classical conserved quantities}
Next we compute the evolution of classical angular momentum and center of mass for  quadrupole moments.
\begin{lemma} Suppose $f_{ij}$ and $g_{ij} $ are both symmetric, traceless $3 \times 3$ matrices. Then
\begin{align}
&\int_{S^2}  (f_{ij} \tilde{X}^i \tilde{X}^j)^2=\frac{8\pi}{15}\sum_{ij} f^2_{ij} \label{int f^2}\\
&\int_{S^2} \tilde{X}^p \epsilon^{AB}\nabla_A (f_{ij} \tilde{X}^i\tilde{X}^j)\nabla_B (g_{kl} \tilde{X}^k \tilde{X}^l)=\frac{16\pi}{15} \sum_j f_{ij} g_{jk} \epsilon^{ikp}.\label{int X{f,g}}
\end{align}
\end{lemma}

\begin{proof}

Both formulae follow from Lemma 5.3 of \cite{CWY_small}
\[\int_{S^2} \tilde{X}^i\tilde{X}^j\tilde{X}^k\tilde{X}^l =\frac{4\pi}{15}(\delta_{ij}\delta_{kl}+\delta_{ik}\delta_{jl}+\delta_{il}\delta_{jk}). \]
\end{proof}

Combining the above lemma with Theorem \ref{evolution2_bracket} and Proposition \ref{energy_momenum_evol_potential}, we conclude that
\begin{proposition} Suppose 
${c}=\sum c_{ij}(u) \tilde{X}^i \tilde{X}^j,  \underline{c}=\sum \underline{c}_{ij}(u) \tilde{X}^i \tilde{X}^j$ ${n}=\sum n_{ij}(u) \tilde{X}^i \tilde{X}^j,  \underline{n}=\sum \underline{n}_{ij}(u) \tilde{X}^i \tilde{X}^j$, then
\begin{equation}\begin{split}
\partial_u E&=-\frac{8\pi}{5} (\sum_{ij} n_{ij}^2+\sum_{ij} \underline{n}_{ij}^2)\\
\partial_u P^k&=-\frac{32\pi}{15} \sum n_{ij} \underline{n}_{jp} \epsilon^{ipk}\\
\partial_u \tilde J^k&= \frac{16\pi}{5} \sum (c_{ij} {n}_{jp}+ \underline{c}_{ij} \underline{n}_{jp}) \epsilon^{ipk}\\
\partial_u \tilde C^k&= \frac{32u\pi}{15} \sum n_{ij} \underline{n}_{jp} \epsilon^{ipk}.
\end{split}\end{equation}

\end{proposition}

\subsection{CWY angular momentum and center of mass}
Next we compute the evolution of the CWY angular momentum and center of mass for  quadrupole moments. We need the following lemma. 
\begin{lemma} \label{news_square_mode} Suppose the potentials of the news tensor are of mode $\ell=2$. Namely, \[N_{AB}=\nabla_A\nabla_B n-\frac{1}{2} \Delta n \sigma_{AB} + \frac{1}{2}(\epsilon_{AC} \nabla_B \nabla^C \underline{n}+\epsilon_{BC} \nabla_A \nabla^C \underline{n})\] where $n=\sum_{ij} n_{ij} \tilde{X}^i \tilde{X}^j$ and $\underline{n}=\sum_{ij}\underline{n}_{ij} \tilde{X}^i \tilde{X}^j$ satisfy $\sum_i n_{ii}=\sum_i \underline{n}_{ii}=0$. Introduce two $\ell=2$ spherical harmonics
\begin{align*}
Q &=\sum_{i, k,l} (n_{ik}n_{il} \tilde{X}^k \tilde{X}^l)-\frac{1}{3}\sum_{i, j}n_{ij}^2,\\
\underline{Q} &= \sum_{i, k,l} (\underline{n}_{ik}\underline{n}_{il} \tilde{X}^k \tilde{X}^l)-\frac{1}{3}\sum_{i, j}\underline{n}_{ij}^2.
\end{align*}
Then
\begin{enumerate}
\item the $\ell=2$ component of $N_{AB}N^{AB}$ is 
\[-\frac{48}{7} Q-\frac{48}{7}\underline{Q},\]
\item the odd mode component of $N_{AB}N^{AB}$ is 

\[ 8 \epsilon^{ik}_{\,\,\,\,m} n_{ij} \underline{n}_{kl} \tilde{X}^m (\delta^{jl}-\tilde{X}^j\tilde{X}^l).\]
Here $\epsilon_{ijk}$ is the Levi-Civita symbol in three dimensions. 
\end{enumerate}
%(3) \[\int_{S^2} \tilde{X}^p N_{AB} N^{AB}=\frac{128\pi}{15} \epsilon^{ik}_{\,\,\,\, p} n_{il} \underline{n}_{kl}.\]
\end{lemma}

\begin{proof} Note that the even-mode components ($\ell=0, 2, 4$) and odd-mode components ($\ell=1,3$) of $N_{AB} N^{AB}$ are given by \[\nabla_A\nabla_Bn \nabla^A\nabla^B n-\frac{1}{2}(\Delta n)^2+\nabla_A\nabla_B \underline{n} \nabla^A\nabla^B \underline{n}-\frac{1}{2}(\Delta \underline{n})^2 \]
and
\begin{align*}
&( \na^A\na^B n - \frac{1}{2}\Delta n \sigma^{AB} ) (\epsilon_{AC} \nabla_B \nabla^C \underline{n}+\epsilon_{BC} \nabla_A \nabla^C \underline{n})\\
&= 2n_{ij}\na^A \tilde X^i \na^B \tilde X^j (\epsilon_{AC} \nabla_B \nabla^C \underline{n}+\epsilon_{BC} \nabla_A \nabla^C \underline{n})\\
&= 2n_{ij}\na^A \tilde X^i \na^B \tilde X^j \cdot 2 \underline{n}_{kl} (\epsilon_{AC}\na_B \tilde X^k \na^C \tilde X^l + \epsilon_{BC} \na_A \tilde X^k \na^C \tilde X^l)\\
&= 8 \epsilon^{ik}_{\,\,\,\,m} n_{ij} \underline{n}_{kl} \tilde{X}^m (\delta^{jl}-\tilde{X}^j\tilde{X}^l)
\end{align*}
respectively. In the last equality we use the identity $\epsilon_{AB}\na_A \tilde X^i \na_B \tilde X^j = \epsilon_{ijk} \tilde X^k$ and $\na_B \tilde X^i \na^B \tilde X^j = \delta^{ij} - \tilde X^i \tilde X^j$. 

For (1), we compute
\[ \nabla_A\nabla_Bn \nabla^A\nabla^B n=20 n^2-8Q+\frac{4}{3}\sum_{ij} n_{ij}^2 .\]

Since the space of $\ell=4$ spherical harmonics is spanned by
\begin{align}
\begin{split}
&\tilde X^i \tilde X^j \tilde X^k \tilde X^l +\frac{1}{35} \lt( \delta^{ij}\delta^{kl}+ \delta^{ik}\delta^{jl}+\delta^{il}\delta^{jk} \rt)\\&- \frac{1}{7}\lt( \tilde X^i \tilde X^j \delta^{kl}+ \tilde X^i \tilde X^k \delta^{jl} + \tilde X^i \tilde X^l \delta^{jk} + \tilde X^j \tilde X^k \delta^{il} + \tilde X^j \tilde X^l \delta^{ik} + \tilde X^k \tilde X^l \delta^{ij} \rt),
\end{split} \end{align}
the $\ell=2$ component of $n^2$ is $\frac{4}{7}Q$. Putting these together, we obtain (1).

\end{proof}

In the case of quadrupole moments, the CWY angular momentum and center of mass take the form:
\begin{equation}\label{CWY_angmom}
J^k=\int_{S^2}    \epsilon^{AB}\nabla_{B}\tilde{X}^{k}[N_A-\frac{1}{4}C_{A}^{\,\,\,\,D}\nabla^B C_{DB}-c\nabla_A m]
\end{equation}
\begin{equation}\label{CWY_com}
C^k = \int_{S^2} \nabla^{A}\tilde{X}^{k}[N_A- u \nabla_A m-\frac{1}{4}C_{A}^{\,\,\,\,D}\nabla^B C_{DB} - \frac{1}{16}\nabla_{A}(C_{DE}C^{DE})-2 \underline{c} \,\,\epsilon_{AB}\nabla^B m],
\end{equation}

Therefore,
\[\begin{split} J^k&=\tilde J^k-\int_{S^2} \tilde{X}^k \epsilon^{AB} \nabla_A c\nabla_B \hat{m}\\
C^k&=\tilde C^k+2 \int_{S^2} \tilde{X}^k \epsilon^{AB} \nabla_A \underline{c}\nabla_B \hat{m}+\frac{1}{4} \int_{S^2} \tilde{X}^k \epsilon^{AB} \nabla_A \underline{c}\nabla_B \Delta(\Delta+2) c \end{split},\] 
where we use \eqref{modified_mass} and \eqref{ibp1}.

The evolution formulae for ${J}^k$ and ${C}^k$ are thus 
\[\begin{split}\partial_u{J}^k=&\partial_u \tilde J^k-\int_{S^2} \tilde{X}^k \epsilon^{AB} \nabla_A n\nabla_B \hat{m}- \int_{S^2} \tilde{X}^k \epsilon^{AB} \nabla_A c \nabla_B \partial_u \hat{m}\\
\partial_u {C}^k=&\partial_u \tilde C^k+2 \int_{S^2} \tilde{X}^k \epsilon^{AB} \nabla_A \underline{n}\nabla_B \hat{m} + 2 \int_{S^2} \tilde{X}^k \epsilon^{AB} \nabla_A \underline{c}\nabla_B \partial_u \hat{m}\\
&+\frac{1}{4} \int_{S^2} \tilde{X}^k \epsilon^{AB} \nabla_A \underline{n}\nabla_B \Delta(\Delta+2) c +\frac{1}{4} \int_{S^2} \tilde{X}^k \epsilon^{AB} \nabla_A \underline{c}\nabla_B \Delta(\Delta+2) n
\end{split}\]

By Lemma \ref{same_mode}, only the $\ell=2$ mode components of $\hat{m}$ and $\partial_u\hat{m}$ will survive in the above integrals. 

Denote the $\ell=2$ mode of $\hat{m}$ by $\hat{m}_{\ell=2}=\hat{m}_{kl} \tilde{X}^k \tilde{X}^l$. By Lemma \ref{news_square_mode}, we get \[\partial_u \hat{m}_{kl}=\frac{6}{7} \lt[ \sum_{i } (n_{ik}n_{il}+\underline{n}_{ik}\underline{n}_{il}) -\frac{1}{3} \delta_{kl} \sum_{i, j}(n_{ij}^2 +\underline{n}_{ij}^2) \rt]. \]
By \eqref{int X{f,g}}, we obtain the evolution equation of $ J^k$ and $ C^k$: 
\begin{proposition} Suppose 
${c}=\sum c_{ij}(u) \tilde{X}^i \tilde{X}^j,  \underline{c}=\sum \underline{c}_{ij}(u) \tilde{X}^i \tilde{X}^j$, ${n}=\sum n_{ij}(u) \tilde{X}^i \tilde{X}^j,  \underline{n}=\sum \underline{n}_{ij}(u) \tilde{X}^i \tilde{X}^j$, then
\begin{equation}\begin{split}
\partial_u E=&-\frac{8\pi}{5} (\sum_{ij} n_{ij}^2+\sum_{ij} \underline{n}_{ij}^2)\\
\partial_u P^k=&-\frac{32\pi}{15} \sum_{i,j,p} n_{ij} \underline{n}_{jp} \epsilon^{ipk}\\
\partial_u {J}^k=& \frac{16\pi}{15} \sum_{i, j, p} (3c_{ij} {n}_{jp}+ 3\underline{c}_{ij} \underline{n}_{jp}-n_{ij} \hat{m}_{jp}- {c}_{ij} \partial_u\hat{m}_{jp} ) \epsilon^{ipk}\\
\partial_u {C}^k=&\frac{16\pi}{15} \sum_{i, j, p} (   2 \underline{n}_{ij} \hat{m}_{jp} +2      \underline{c}_{ij} \partial_u\hat{m}_{jp} +6\underline{n}_{ij} {c}_{jp}+ 6\underline{c}_{ij} {n}_{jp}) \epsilon^{ipk} \\
&+\frac{32u\pi}{15} \sum n_{ij} \underline{n}_{jp} \epsilon^{ipk}, 
\end{split}\end{equation} where $\hat{m}_{kl}$ is given by 
 \[\partial_u \hat{m}_{kl}=\frac{6}{7} [\sum_{i } (n_{ik}n_{il}+\underline{n}_{ik}\underline{n}_{il}) -\frac{1}{3} \delta_{kl} \sum_{i, j}(n_{ij}^2 +\underline{n}_{ij}^2) ] \]

\end{proposition}

\appendix
\begin{comment}
\section{Commutation formulae}
We recall the curvature formula: 
\[\nabla_A\nabla_B V_C-\nabla_B\nabla_A V_C=R_{ABC}^{\qquad D} V_D\]
\[\nabla_A\nabla_B V_{CD}-\nabla_B\nabla_A V_{CD}=R_{ABC}^{\qquad E} V_{ED}+ R_{ABD}^{\qquad E} V_{CE}\]

Specializing to the 2-sphere case, the curvature is 

\[R_{ABC}^{\qquad D}=\sigma_{AC} \delta_B^{\;\; D}-\delta_A^{\;\;D}\sigma_{BC}.\]

 \[\epsilon_A^{\,\,\,\, C} \epsilon^{BD}=\delta_A^{\;\; B} \sigma^{CD}-\delta_A^{\;\;D}\sigma^{CB}=R_{A}^{\;\; CBD}\]
 
 Therefore
 
 \[\nabla_A\nabla_B V_C-\nabla_B\nabla_A V_C=\sigma_{AC} V_B-\sigma_{BC} V_A\]
 
 \[\nabla_A\nabla_B V_{CD}-\nabla_B\nabla_A V_{CD}=\sigma_{AC} V_{BD}-\sigma_{BC} V_{AD}+ \sigma_{AD} V_{CB}- \sigma_{BD} V_{CA}.\]
or 

\begin{equation}\label{curvature1}\epsilon^{AB}\nabla_A\nabla_B V_{CD}=\epsilon_{C}^{\;\;B} V_{BD}+\epsilon_D^{\;\; B} V_{BC}\end{equation}
\end{comment}

\section{Christodoulou-Klainerman connection coefficients and curvature components in Bondi-Sachs formalism}
We write the limit of connection coefficients and curvature components defined in \cite{CK, Christo1991, ChristoMG9} in terms of the Bondi-Sachs metric coefficients.

We choose the null vector fields $L = \frac{\pl}{\pl r}$ and $\Lb = \frac{2}{U}\lt( \pl_u - W^D \pl_D - \frac{V}{2}\pl_r \rt)$, which satisfy $\langle L,\Lb \rangle = -2$.

\begin{definition}
The second fundamental forms and torsion are defined by
\begin{align*}
\chi_{AB} &= \langle D_A L, \pl_B \rangle = \frac{1}{2} \tr\chi g_{AB} + \hat\chi_{AB}\\
\chib_{AB} &= \langle D_A \Lb, \pl_B \rangle = \frac{1}{2} \tr\chib g_{AB} + \hat\chib_{AB}\\
\zeta_A &= \frac{1}{2} \langle D_A L, \Lb \rangle
\end{align*}
Their limit as $r \rw\infty$ are defined by
\begin{align*}
\Sigma &= \lim_{r\rw\infty} \hat\chi\\
\Xi &= \lim_{r\rw\infty} r^{-1} \hat\chib \\
Z &= \lim_{r\rw\infty} r\zeta .
\end{align*}
\end{definition}
They are related to the metric coefficients in the corresponding Bondi-Sachs coordinate system as follows:
\begin{prop}\label{tensor in BS variable}
\begin{align*}
\Sigma_{AB} &= -\frac{1}{2} C_{AB}\\
\Xi_{AB} &= N_{AB}\\
Z_A &= -\frac{1}{2} \na^B C_{AB}
\end{align*}
\end{prop}
\begin{proof}
Starting with $g_{AB} = r^2 \sigma_{AB} + rC_{AB} + O(1)$, the determinant condition gives $\tr\chi = \frac{2}{r}$ and we compute
\begin{align*}
\chi_{AB} = r\sigma_{AB} + \frac{1}{2}C_{AB} + O(r^{-1})
\end{align*}
to get $\Sigma_{AB} = -\frac{1}{2}C_{AB}$. Direct computation gives
\begin{align*}
\chib_{AB} = r(-\sigma_{AB} + \pl_u C_{AB} ) + O(1)
\end{align*}
and hence $\tr\chib = -\frac{2}{r} + O(r^{-2})$ and $\hat\chib_{AB} = r \pl_u C_{AB} + O(1)$. The limit of torsion follows from $\zeta_A = -\frac{1}{r} W_A^{(-2)} + O(r^{-2})$.
\end{proof}

\begin{definition}
The mass aspect and conjugate mass aspect function of Christodoulou-Klainerman are defined by
\begin{align*}
\mu = K + \frac{1}{4}\tr\chi\tr\chib - \mbox{div}\zeta \qquad \underline{\mu} = K + \frac{1}{4}\tr\chi\tr\chib + \mbox{div}\zeta.
\end{align*}
Here $K$ denotes the Gauss curvature of the two-sphere $r=$const. Their limits are defined by
\begin{align*}
N &= \lim_{r\rw\infty} r^3 \mu \\
\underline{N} &= \lim_{r\rw\infty} r^3 \underline{\mu}.
\end{align*}
\end{definition}
We express them in terms of the Bondi-Sachs metric coefficients as follows:
\begin{prop}
\begin{align*}
N = 2m + \frac{1}{2}\na^A\na^B C_{AB}, \qquad \underline{N} = 2m - \frac{1}{2}\na^A\na^B C_{AB}
\end{align*}
\end{prop}
\begin{proof}
We compute $K = \frac{1}{r^2} + \frac{1}{2r^3} \na^A\na^B C_{AB} + O(r^{-4})$ and $\frac{1}{4}\tr\chi\tr\chib = -\frac{1}{r^2} + \frac{1}{r^3}(2m-\frac{1}{2}\na^A\na^B C_{AB})$ and the assertion follows.
\end{proof}

We turn to curvature components. The convention of Riemann curvature tensor is
\begin{align*}
R(X,Y)Z &= (D_X D_Y -D_Y D_X - D_{[X,Y]} )Z\\
R(X,Y,W,Z) &= \langle R(X,Y)Z,W \rangle.
\end{align*}
\begin{definition}\label{def curv components}
Define the curvature components
\begin{align*}
\underline{\alpha}_{AB} &=  R(\pl_A, \Lb, \pl_B, \Lb)\\
\underline{\beta}_A &= \frac{1}{2} R(\pl_A, \Lb,\Lb, L)\\
\rho &= \frac{1}{4} R(\Lb,L,\Lb,L)\\
\sigma \slashed{\epsilon}_{AB} &= \frac{1}{2} R(\pl_A, \pl_B, \Lb, L)\\
\beta_A &= \frac{1}{2}R(\pl_A,L,\Lb,L)
\end{align*}
Here $\slashed{\epsilon}_{AB} dx^A \wedge dx^B$ is the area form of the two-sphere with respect to $g_{AB}$. Their limits are defined by
\begin{align*}
\underline{A}_{AB} &= \lim_{r\rw\infty} r^{-1}\underline{\alpha}_{AB}\\
\underline{B}_A &= \lim_{r\rw\infty} r\underline{\beta}_A\\
P &= \lim_{r\rw\infty} r^3\rho\\
Q&= \lim_{r\rw\infty}r^3 \sigma\\
B_A &= \lim_{r\rw\infty} r^3 \beta_A
\end{align*}
Note that $(\underline{A},\underline{B})$ were denoted by $(A,B)$ in \cite{Christo1991}.
\end{definition}
We express them in terms of the Bondi-Sachs metric coefficients as follows:
\begin{prop}\label{curv in BS variable}
\begin{align*}
\underline{A}_{AB} &= -2 \pl_u N_{AB}\\
\underline{B}_A &= \na^B N_{AB}\\
P &= -2m - \frac{1}{4}C_{AB}N^{AB}\\
Q&= \epsilon^{AB} \lt( -\frac{1}{4}C_A^D N_{DB} - \frac{1}{2} \na_A \na^D C_{DB} \rt)\\
B_A &= -N_A 
\end{align*}
\end{prop}
\begin{proof}
The formula for $\underline{A}$ is obtained from (6) of \cite{Christo1991}, $2 \frac{\pl \Xi}{\pl u} = -\underline{A}$, which is the rescaled limit of the propagation equation $\hat{\underline{D}} \hat\chib = -\underline{\alpha}$.

The formula for $\underline{B}$ is obtained from (2) of \cite{Christo1991}, $\na^B \Xi_{AB} = \underline{B}_A$, which is the rescaled limit of the Codazzi equation $\slashed{\mbox{div}} \hat\chib - \hat\chib \cdot \zeta = \frac{1}{2} \lt( \slashed{\na} \tr\chib - \tr\chib \zeta \rt) + \underline{\beta}$.

The formula for $P$ and $Q$ are obtained from (3) of \cite{Christo1991},
\[ \epsilon^{AB} \na_A Z_B = Q - \frac{1}{2}\Sigma \wedge \Xi, \qquad \na^A Z_A = \underline{N} + P - \frac{1}{2} \Sigma \cdot \Xi, \]
which is the rescaled limit of the Hodge system
\[ \slashed{curl} \zeta = \sigma -\frac{1}{2} \hat\chi \wedge \hat\chib, \qquad \slashed{\mbox{div}} \zeta = \underline{\mu} + \rho -\frac{1}{2}\hat\chi \cdot \hat\chib. \]

Finally, we consider the Codazzi equation
\[ \slashed{\mbox{div}} \hat\chi + \hat\chi \cdot \zeta = \frac{1}{2} \lt( \slashed{\na} \tr\chi + \tr\chi \zeta \rt) - \beta. \]
Its leading order at $O(r^{-2})$ leads to (1) of \cite{Christo1991} and its subleading order at $O(r^{-3})$ leads to
\[\begin{split} &\lt( -\frac{1}{4} \pl_A |C|^2 + \frac{1}{2} C^{BD} \na_D C_{AB} + \frac{1}{4} \na_A C_D^E C^D_E + \frac{1}{2} \na_D C^{DE}C_{AE} \rt) + \frac{1}{4}C_{AB} \na_D C^{BD}\\
 = &\zeta_A^{(-2)} - B_A. \end{split} \]
We simplify the second term in the parentheses by the identity $\na_{(D}C_{B)A} = \na_A C_{BD} + \na^E C_{AE} \sigma_{BD} - \na^E C_{E(D} C_{B)A}$ and the left-hand side becomes $\frac{1}{8} \pl_A |C|^2 + \frac{1}{4}C_{AB}\na_D C^{BD}$. Direct computation yields $\zeta_A^{(-2)} = -N_A + \frac{1}{8} \pl_A |C|^2 + \frac{1}{4}C_{AB}\na_D C^{BD}$ and the formula for $B$ follows.
\end{proof}
\section{Integration by Part Formula}
In this section, we prove two integration formula that will be used to compute angular momentum and center of mass in spacetime with vanishing news.
\begin{theorem}\label{thm_int_1}
Let $F_{AB} = 2 \na_A\na_B f - \Delta f \sigma_{AB}$ and $P_{BA} = \na_B\na^D C_{DA} - \na_A\na^D C_{DB}$. Then
\begin{align*}
\int_{S^2} Y^A \lt( \frac{1}{4} C_{AB} \na_D F^{DB} + \frac{1}{4}	 F_{AB} \na_D C^{DB} -\frac{3}{4} P_{BA} \na^B f -\frac{1}{4} \na^B P_{BA} f \rt)=0.
\end{align*}
\end{theorem}
\begin{proof}
We integrate by parts the last two terms to get 
\begin{align*}
&\int_{S^2} -\frac{1}{2} Y^A  (\na_B\na^D C_{AD} - \na_A\na^D C_{BD}) \na^B f + \frac{1}{2}\na^B Y^A \na_B\na^D C_{AD} \cdot f.\\
=& \int_{S^2} \frac{1}{2} \na_B Y^A \na^D C_{AD} \na^B f + \frac{1}{2} Y^A \na^D C_{AD} \Delta f - \frac{1}{2} Y^A \na^D C^B_D \na_A\na_B f\\
& + \int_{S^2} \frac{1}{2} Y^A \na^D C_{AD} f - \frac{1}{2} \na^B Y^A \na^D C_{AD} \na_B f\\
= & \int_{S^2} -\frac{1}{2} Y^A \na^D C_D^B (\na_A\na_B f - \frac{1}{2}\Delta f \sigma_{AB}) + \frac{1}{4} Y^A \na^D C_{AD} (\Delta + 2)f
\end{align*}
\end{proof}

\begin{theorem}\label{thm_int_2}
Let $F_{AB} = 2 \na_A\na_B f - \Delta f \sigma_{AB}$ and $P_{BA} = \na_B\na^D C_{DA} - \na_A\na^D C_{DB}$. Then 
\begin{align*}
\int_{S^2} \na^A \tilde X^k \lt( C_{AB} \na_D F^{BD} + F_{AB}\na_D C^{BD} + \frac{1}{2} \na_A (C_{BD}F^{BD} ) - f\na^B P_{BA} - 3 P_{BA} \na^B f \rt)=0. 	
\end{align*}
\end{theorem}
\begin{proof}
We integrate by parts the last two terms to get 
\begin{align*}
&\int_{S^2} \na^A \tilde X^k (-2P_{BA})\na^B f\\
=&  \int_{S^2} -2\na^A \tilde X^k (\na_B\na^D C_{DA} - \na_A\na^D C_{DB})\na^B f\\
=& \int_{S^2} -2\tilde X^k \na^D C_{DA}\na^A f + 2 \na^A \tilde X^k \na^D C_{DA}\Delta f + 4 \tilde X^k \na^D C_{DB}\na^B f - 2\na^A \tilde X^k \na^D C_{DB} \na_A\na^B f\\
=& \int_{S^2} 2 \tilde X^k \na^D C_{DA}\na^A f -\na^A \tilde X^k \na^D C_{DB} F_{AB} + \na^A \tilde X^k \na^D C_{DA}\Delta f \\
=& \int_{S^2} -2 \na^D \tilde X^k C_{DA}\na^A f -2 \tilde X^k C_{DA} \na^D\na^A f -\na^A \tilde X^k \na_D C^{DB} F_{AB} + \na^A \tilde X^k \na^D C_{DA}\Delta f \\
=& \int_{S^2} 2\na^A \tilde X^k \na^D C_{DA} \na^A f + \frac{1}{2}\Delta\tilde X^k C_{DA}F^{DA} - \na^A \tilde X^k \na_D C^{DB} F_{AB} + \na^A \tilde X^k \na^D C_{DA}\Delta f\\
=& \int_{S^2} -\na^A \tilde X^k C_{DA}\na^D(\Delta + 2)f -\frac{1}{2}\na^A \tilde X^k \na_A (C_{DB}F^{DB}) - \na^A \tilde X^k \na_D C^{DB} F_{AB}.
\end{align*}
\end{proof}

\end{document}